\documentclass[a4paper,fleqn]{cas-sc}
\usepackage{longtable}
\usepackage{enumitem}
\usepackage{appendix}

\usepackage[authoryear]{natbib}
\usepackage{color}
\usepackage{amsmath}
\usepackage{amsthm}
\usepackage[linesnumbered,ruled]{algorithm2e}
\usepackage{bm}

\def\tsc#1{\csdef{#1}{\textsc{\lowercase{#1}}\xspace}}
\tsc{WGM}
\tsc{QE}
\tsc{EP}
\tsc{PMS}
\tsc{BEC}
\tsc{DE}

\newtheorem{proposition}{Proposition}
\newtheorem{definition}{Definition}

\newtheorem{theorem}{Theorem}
\newtheorem{lemma}{Lemma}
\newtheorem{corollary}{Corollary}
\newtheorem{assumption}{Assumption}

\begin{document}
\let\WriteBookmarks\relax
\def\floatpagepagefraction{1}
\def\textpagefraction{.001}

\shorttitle{Stability-Preserving Com-MP}
\shortauthors{Zhu and Tan}

\title [mode = title]{\textbf{Stability-Preserving Reinforcement Learning for Heterogeneity-Aware Composite Max-Pressure Traffic Signal Control}}

\author[1]{Jichen Zhu}[orcid=0009-0003-1271-4166]
\ead{jichen.zhu@nus.edu.sg}
\credit{Writing - review editing, Writing - original draft, Visualization, Validation, Software, Methodology, Resources}

\author[2]{Chaopeng Tan}[orcid=0000-0003-4737-5304]
\ead{chaopeng.tan@tu-dresden.de}
\cormark[1]
\credit{Writing - review editing, Writing - original draft, Methodology, Investigation, Formal analysis, Resources, Supervision, Conceptualization}

\affiliation[1]{organization={Department of Civil and Environmental Engineering, National University of Singapore}, 
        city={Singapore},
        adressline={1 Engineering Drive 2,},
        postcode={117576}, 
        country={Singapore}}

\affiliation[2]{organization={Chair of Traffic Process Automation, Technische Universit\"at Dresden},
    addressline={Hettnerstra\ss e 3},
    city={Dresden},
    postcode={01069},
    country={Germany}}
\cortext[cor1]{Corresponding author}

\begin{abstract}
Reinforcement learning (RL) provides a flexible mechanism for adapting traffic signal control to heterogeneous traffic and sensing conditions, but directly learned signal actions do not automatically inherit the analytical properties of structured traffic-control methods. This paper proposes Com-MP, a theory-guided RL framework that learns a context-dependent composition of multiple max-pressure (MP) control preferences while retaining decentralized MP phase selection. The learned policy adapts the relative influence of constituent MP formulations according to local traffic and sensing conditions, allowing data-driven adaptation to operate within a structured traffic-control framework. We show that such intersection-specific and time-varying composition preserves strong queue stability for every demand in the strictly stabilizable region under any admissible non-anticipative composition policy, thereby separating the performance-adaptation role of learning from the stability property of the underlying controller. In a representative implementation combining detector-supported queue-based MP and connected-vehicle-supported travel-time MP, experiments on a real-world 45-intersection network show that Com-MP reduces average delay by 15.6\% and 8.6\% relative to the two constituent controllers, respectively, and by about 7.0\% relative to the best fixed composition. The learned composition responds systematically to local traffic and sensing conditions, remains effective under heterogeneous and degraded sensing, and can be transferred without retraining to an unseen network with an average delay only 2.4\% higher than that of a policy trained directly on the target network. These results show how learning-based adaptation can be introduced within a stability-preserving traffic-control structure, providing a theory-guided pathway toward trustworthy learning-enabled traffic signal control.\footnote{The source code supporting the findings of this study is openly available on GitHub at \url{https://github.com/bryan-zjc/Com-MP}.}
\end{abstract}

\begin{keywords}
Theory-guided reinforcement learning \sep Max-pressure control \sep Queue stability \sep Heterogeneous traffic sensing \sep Trustworthy artificial intelligence
\end{keywords}

\maketitle

\section{Introduction}

Real-time traffic signal control at the network level remains challenging because urban traffic demand is stochastic, spatially coupled, and continuously changing. Centralized control can explicitly optimize network-wide objectives, but it generally requires large-scale prediction and high-dimensional optimization, creating substantial computational, communication, and modeling burdens as the network grows \citep{yan2019network,wang2020optimizing,zhu2025asynchronous}. Max-pressure (MP) control offers a fundamentally different approach. Each intersection selects its signal state using traffic information from adjacent movements, while the resulting decentralized policy can strongly stabilize network queues and achieve maximum throughput for demands in the admissible region \citep{varaiya2013max}. This combination of local implementation, scalability, and network-level stability has made MP an important foundation for large-scale adaptive signal control. Its practical behavior, however, depends critically on how traffic pressure is represented and on the sensing information available to construct that representation.

MP encompasses a family of policies that share the pressure-maximization principle but represent traffic congestion in different ways. Classical queue-based MP uses vehicle accumulation or link occupancy \citep{varaiya2013max}, while position-weighted formulations account for the spatial distribution of vehicles and downstream storage conditions \citep{li2019position}. Delay- and travel-time-based formulations introduce temporal traffic experience through head-of-line delay, cumulative delay, or vehicle travel time \citep{wu2018delay,mercader2020max,liu2022novel,liu2023total,tan2026cvmp}. Other extensions incorporate switching losses, pedestrian demand, transit priority, and related operational considerations \citep{wang2022learning,liu2024max,ahmed2024occ,tan2025transit}. These formulations encode different control preferences. Queue-based pressure responds directly to vehicle accumulation and tends to prioritize heavily loaded movements, whereas delay- and travel-time-based states place increasing emphasis on traffic that has accumulated temporal cost. Their relative suitability can therefore change with the local congestion pattern rather than being determined by a single universally preferable pressure definition.

The usefulness of a pressure formulation also depends on the sensing environment supporting its traffic state. Fixed-location detectors provide persistent aggregate observations at instrumented locations, but practical networks may have incomplete coverage and may experience sensing or communication failures \citep{klein2006traffic}. Detector-supported states at unmonitored locations consequently rely on reconstruction from available counts, traffic-flow conservation, turning information, or related estimation procedures \citep{tan2026privatemp}. Connected vehicles (CVs) provide a complementary sensing source by supplying vehicle-level trajectory and travel-time information beyond fixed detector locations \citep{tan2021cumulative,wang2024traffic,tan2025connected_robust,zhang2026integrated,zhu2025joint}. Their effective observation coverage, however, is penetration-dependent and spatially heterogeneous \citep{wong2019estimation,tan2026cvmp,tan2026privatemp}. Consequently, the information supporting a pressure formulation can become more or less representative as detector availability, CV penetration, and local traffic conditions change.

Traffic-state heterogeneity and sensing heterogeneity therefore interact. A queue-oriented formulation may be particularly informative under strong local accumulation when reliable queue observations are available, whereas a trajectory-based formulation can contribute complementary temporal information when sufficient vehicle-level observations are present. Their relative usefulness can also vary across intersections because congestion patterns and sensing availability are spatially heterogeneous. Most existing MP studies improve a particular pressure formulation and then apply that formulation with a fixed control preference. Heterogeneous sensing can also be fused to construct a more complete estimate of a chosen traffic state; for example, detector and CV information can be integrated to support travel-time-based pressure estimation \citep{tan2026privatemp}. Such approaches improve how a particular pressure state is constructed. A different question arises when several valid MP formulations coexist: \emph{how should their relative influence be adapted as the local traffic state and sensing environment change?}

Reinforcement learning (RL) provides a natural mechanism for such adaptation by learning decision rules through repeated interaction with dynamic traffic environments. RL-based traffic signal control has consequently developed into a broad research area spanning direct phase selection, multi-intersection coordination, and network-level control \citep{noaeen2022rlreview,han2023rltraffic,zhu2026schedule}. Representative approaches learn signal actions or coordination mechanisms directly, with traffic pressure often incorporated into the state, reward, or learning architecture \citep{wei2019presslight,wei2019colight}. At the same time, increasing attention has been given to incorporating transportation structure, reference policies, and operational constraints into learning-based control. Examples include adjusting selected components of an MP controller through learning \citep{wang2022learning}, guiding RL using a model-based reference signal policy \citep{lu2025referencerl}, and imposing operational constraints through safe RL formulations \citep{zhou2025saferl}. More broadly, safe and learning-based control research has emphasized combining data-driven adaptation with explicit control structure or constraints so that desirable system properties are not determined solely by the learned policy \citep{garcia2015safe,brunke2022safe}. For traffic signal control, this distinction is particularly relevant to MP: its appeal lies not only in decentralized decision-making but also in its explicit network-level queue-stability properties \citep{varaiya2013max}. When RL directly determines signal actions, however, the learned policy does not automatically inherit these properties unless they are explicitly retained in the control architecture and established analytically.

This gap motivates a theory-guided role for learning in MP control. Rather than replacing the underlying traffic-control mechanism, learning can adapt selected decisions within a structured controller whose analytical properties remain available. For the problem considered here, this leads to a specific question: \emph{can reinforcement learning adapt among heterogeneous MP control preferences while preserving the network-level queue-stability structure of MP independently of the particular learned policy?} 

Such a separation would allow learning to focus on context-dependent performance adaptation, while the lower-level MP mechanism remains responsible for actual traffic regulation and the associated stability properties.
Motivated by this question, we develop \emph{Com-MP}, a theory-guided RL framework for context-dependent composition of multiple MP controllers. At each signal decision epoch, constituent controllers evaluate candidate phases according to their respective pressure formulations, and a shared RL policy adapts their relative influence using local traffic and sensing observations. Learning therefore determines \emph{which MP control preference should be emphasized} under the prevailing context, while decentralized MP remains responsible for signal-phase selection. In the representative implementation, Com-MP combines a detector-supported queue-based MP controller and a CV-supported travel-time-based MP controller, which differ in both control preference and sensing source. The framework is not restricted to these two constituents and can accommodate a finite collection of admissible MP formulations.

A central theoretical question is whether such spatially and temporally varying composition preserves the stability properties that motivate MP control. We show that adaptive composition preserves strong queue stability for every demand in the strictly stabilizable region under any admissible non-anticipative composition policy. The guarantee therefore does not depend on the particular RL policy or training algorithm used to generate the composition within the proposed framework. This establishes a separation between learning and stability: learning determines how the controller adapts its relative control preferences, whereas network-level queue stability follows from the structured MP mechanism. The detector-supported Q-MP and CV-supported travel-time MP states used in the representative implementation are further shown to satisfy the conditions of the general stability result.

The main contributions of this paper are threefold:
\begin{itemize}

    \item We develop a theory-guided RL architecture for trustworthy traffic signal control by adapting heterogeneous MP control preferences within a structured decision-making framework. A shared policy learns intersection-specific and time-varying compositions from local traffic and sensing conditions, while decentralized MP remains responsible for actual signal-phase selection. This design confines learning to an interpretable adaptation layer, allowing data-driven flexibility without replacing the underlying control mechanism.

    \item We establish a general strong-stability result for adaptive composite MP control. For a finite collection of admissible constituent MP states, arbitrary non-anticipative, intersection-specific, and time-varying composition policies preserve strong queue stability for every demand in the strictly stabilizable region under the stated finite-storage network dynamics and operational signal constraints. The stability guarantee is therefore independent of the particular RL policy used to select the composition, providing a formal network-level guarantee for the proposed learning-enabled control architecture.

    \item We evaluate the representative Q--CV Com-MP under heterogeneous traffic and sensing conditions and examine its cross-network transferability. In the real-world 45-intersection Qinzhou network, Com-MP reduces average delay by 15.6\% and 8.6\% relative to Q-MP and CV-MP, respectively, and by about 7.0\% relative to the best fixed composition. The learned composition varies systematically with local traffic and sensing conditions, remains effective under heterogeneous and degraded sensing, and transfers without retraining to an unseen real-world network with an average delay only 2.4\% higher than that of a policy trained directly on the target network. These results show that the learned composition responds meaningfully to local conditions, remains effective when sensing degrades, and can be reused on an unseen network.

\end{itemize}

The remainder of this paper is organized as follows. Section~\ref{sec:problem_setting} introduces the network model and the constituent MP controllers. Section~\ref{sec:com_mp} presents the composite MP formulation and establishes its stability properties. Section~\ref{sec:rl_composition} develops the theory-guided learning mechanism for context-dependent composition. Section~\ref{sec:evaluation} evaluates traffic performance and learned composition behavior, then examines sensing robustness and cross-network transferability. Section~\ref{sec:conclusion} concludes the paper and discusses future research directions.

\section{Problem Setting and Constituent MP Controllers}
\label{sec:problem_setting}

This section introduces the traffic-network model and a common representation for MP controllers with different traffic-state definitions and sensing requirements. We first describe the movement-count dynamics and operational signal constraints used in the stability analysis. We then formulate a generalized MP controller in which different constituent controllers share the same pressure-based phase-selection structure but may use different movement states. Finally, the detector-supported queue-based MP and CV-supported travel-time-based MP controllers used in the representative implementation are introduced.

\subsection{Network model and traffic-flow dynamics}
\label{sec:network_model}

Time is continuous and denoted by $t\ge0$. Signal decisions are made at epochs $t_k$, $k=0,1,\ldots$, with $t_0=0$. Let $I_k=[t_k,t_{k+1})$ denote the $k$th decision interval and let $\Delta=t_{k+1}-t_k>0$ denote its fixed duration.

Consider a signalized road network with a set of real intersections $\mathcal N$ and a set of fictitious source nodes $\mathcal F$. Fictitious nodes represent exogenous demand entering the modeled network and provide point queues in which demand that cannot immediately enter the physical network can accumulate. For each node $n\in\mathcal N\cup\mathcal F$, let $\mathcal M_n$ denote its movement set. A real movement is indexed by $a=(i,j)$, representing traffic traveling from incoming link $i$ to outgoing link $j$. Let $L_i>0$ denote the physical length of real link $i$. Define the real-movement set as $\mathcal M_{\mathcal N}=\bigcup_{n\in\mathcal N}\mathcal M_n$, the fictitious-source movement set as $\mathcal M_{\mathcal F}=\bigcup_{n\in\mathcal F}\mathcal M_n$, and the complete movement set as $\mathcal M=\mathcal M_{\mathcal N}\cup\mathcal M_{\mathcal F}$. Real movements have finite physical storage, whereas fictitious source movements have no physical length and are modeled as point queues with unbounded storage. This distinction allows congestion inside the network and demand waiting outside the network to be represented within one movement-based model \citep{li2019position,tan2026cvmp}.

For each real intersection $n\in\mathcal N$, let $\mathcal P_n$ denote its finite set of compatible signal phases. A phase $\phi\in\mathcal P_n$ corresponds to the movement-level signal vector $\bm s_n^\phi=\{s_{n,a}^{\phi}\}_{a\in\mathcal M_n}$, where $s_{n,a}^{\phi}=1$ if movement $a$ is served by phase $\phi$ and $s_{n,a}^{\phi}=0$ otherwise. Let $\mathcal S_n=\{\bm s_n^\phi:\phi\in\mathcal P_n\}$ denote the local signal-state set and $\mathcal S=\prod_{n\in\mathcal N}\mathcal S_n$ the network signal-state space. A network signal state is denoted by $\bm s=(\bm s_n)_{n\in\mathcal N}\in\mathcal S$, where $\bm s_n$ is its local component at intersection $n$ and $s_a$ is the activation component associated with movement $a$. The only movement at each fictitious source node is treated as continuously available.

To prevent a compatible phase from remaining unserved indefinitely, let $T_{n,\phi}^{red}(t)$ denote the elapsed red time of phase $\phi$ at intersection $n$, and let $T_{n,\phi}^{red,\max}<\infty$ denote its prescribed maximum red time. An admissible signal policy satisfies
\begin{equation}
    T_{n,\phi}^{red}(t)
    \le
    T_{n,\phi}^{red,\max},
    \qquad
    n\in\mathcal N,\quad
    \phi\in\mathcal P_n.
    \label{eq:maximum_red_requirement}
\end{equation}
Let $\Sigma_{\mathrm{adm}}$ denote the class of nonanticipative signal policies satisfying the compatible-phase and maximum-red constraints.

For each movement $a\in\mathcal M$, let $z_a(t)$ denote its vehicle count. For a real movement, $z_a(t)$ is the number of vehicles currently occupying that finite movement; for $a\in\mathcal M_{\mathcal F}$, it denotes the fictitious source queue. Let $\bm z(t)=[z_a(t)]_{a\in\mathcal M}$ collect all movement counts.

For a real movement $a=(i,j)$, let $\mathcal D(a)\subseteq\mathcal M$ denote the set of immediately downstream movements $b=(j,\ell)$, where $\ell$ denotes a link downstream of $j$. For each $b\in\mathcal D(a)$, let $r_{b,k}$ denote the proportion of traffic on link $j$ assigned to downstream movement $b$ during $I_k$, with $r_{b,k}\ge0$ and $\sum_{b\in\mathcal D(a)}r_{b,k}\le1$; any remaining fraction represents traffic leaving the modeled network after link $j$. Define the movement-transition matrix $\bm R_k$ by $(\bm R_k)_{ab}=r_{b,k}$ for $b\in\mathcal D(a)$ and $(\bm R_k)_{ab}=0$ otherwise. The turning proportions are available at $t_k$ and are held fixed over $I_k$; equivalently, write $r_b(t)=r_{b,k}$ and $\bm R(t)=\bm R_k$ for $t\in I_k$.

Let $d_a(t)$ denote the realized discharge rate of movement $a$ and let $\bm d(t)=[d_a(t)]_{a\in\mathcal M}$ collect these rates. Let $\bm\lambda(t)$ denote the exogenous arrival-rate vector, whose entries are zero for real movements and may be positive for fictitious source movements. At the aggregate rate level, movement conservation is represented by
\begin{equation}
    \dot{\bm z}(t)
    =
    \bm\lambda(t)
    -
    [\bm I-\bm R^\top(t)]\bm d(t),
    \qquad
    t\in I_k,
    \label{eq:network_count_dynamics}
\end{equation}
where $\bm I$ is the identity matrix of dimension $|\mathcal M|$. Under an absolutely continuous flow representation, the overdot denotes the ordinary time derivative. Under an event-based vehicle-count representation, individual arrivals and movement transitions produce bounded unit jumps; their conditional contribution is included through the extended drift operator introduced in Section~\ref{sec:stability_conditions}. Both representations follow the same first-order conservation balance in~\eqref{eq:network_count_dynamics}. With the above orientation of $\bm R(t)$, $\bm R^\top(t)\bm d(t)$ represents the endogenous inflow generated by upstream movement discharges. Let $\bar{\bm\lambda}$ denote the stationary mean exogenous demand vector used in the stability analysis; the corresponding conditional-mean condition is stated in Assumption~\ref{ass:network_regularity}.

The discharge that can actually be delivered by a movement depends not only on the signal state but also on the current sending and receiving conditions. Let $d_a(\bm s,t)$ denote the counterfactual discharge rate that would result if candidate network signal state $\bm s\in\mathcal S$ were applied at time $t$. Its conditional expected value is
\begin{equation}
    \bar d_a(\bm s,t)
    \triangleq
    \mathbb E
    \left[
        d_a(\bm s,t)
        \mid
        \mathcal I_{t^-}
    \right],
    \qquad
    0\le\bar d_a(\bm s,t)\le c_a,
    \label{eq:candidate_expected_discharge}
\end{equation}
where $\mathcal I_{t^-}$ denotes the traffic information available immediately before time $t$, and $c_a>0$ is a finite upper bound on the service rate of movement $a$. The effective expected discharge incorporates the current downstream receiving condition and is therefore reduced when downstream spillback restricts service. For a real movement $a\in\mathcal M_n$, candidate discharge depends on $\bm s$ only through its local component $\bm s_n$. For $a\in\mathcal M_{\mathcal F}$, fictitious-source discharge is independent of the candidate real-intersection phases, and we write $\bar d_a(t)$ for this candidate-independent quantity. For a real movement $a\in\mathcal M_n$, $\bar d_a(\bm s_n^\phi,t)$ denotes the same conditional expected discharge evaluated under a network signal state whose local component at intersection $n$ is $\bm s_n^\phi$. Let $\bar{\bm d}(\bm s,t)$ collect the candidate expected discharge rates of all movements. Under the conventional saturation-flow abstraction for a signal-controlled movement, $\bar d_a(\bm s,t)=c_as_a$.

The formulation in~\eqref{eq:network_count_dynamics}--\eqref{eq:candidate_expected_discharge} separates physical traffic conservation from signal-dependent service. Different MP controllers can therefore use different states to represent congestion while acting on the same underlying movement-count dynamics.

\subsection{Generalized MP formulation}
\label{sec:generalized_mp}

Let $\mathcal H=\{1,\ldots,H\}$ index a finite collection of constituent MP controllers. Constituent controller $h\in\mathcal H$ assigns a nonnegative traffic state $w_a^h(t)$ to each real movement $a\in\mathcal M_{\mathcal N}$. Depending on the constituent formulation, this state may represent vehicle accumulation, a spatially weighted congestion measure, or a temporal measure such as delay or travel time. For fictitious source movements, all constituents use the same physical source queue, i.e., $w_a^h(t)=z_a(t)$ for $a\in\mathcal M_{\mathcal F}$ and $h\in\mathcal H$. Let $\bm w^h(t)=[w_a^h(t)]_{a\in\mathcal M}$ collect the movement states of constituent $h$.

For movement $a$, define the downstream differential pressure of constituent $h$ as
\begin{equation}
    p_a^h(t)
    =
    w_a^h(t)
    -
    \sum_{b\in\mathcal D(a)}
    r_b(t)w_b^h(t).
    \label{eq:movement_pressure}
\end{equation}
The first term represents the congestion state associated with movement $a$, whereas the downstream term accounts for the congestion states of the movements that receive its discharge.

The network pressure-service objective of constituent $h$ is
\begin{equation}
    \Pi^h(\bm s,t)
    =
    \sum_{a\in\mathcal M}
    p_a^h(t)\,
    \bar d_a(\bm s,t)
    =
    (\bm w^h(t))^\top
    [\bm I-\bm R^\top(t)]
    \bar{\bm d}(\bm s,t).
    \label{eq:constituent_pressure_matrix}
\end{equation}
The constituent MP rule selects a feasible signal state satisfying
\begin{equation}
    \bm s^{h,*}(t)
    \in
    \arg\max_{\bm s\in\mathcal S}
    \Pi^h(\bm s,t).
    \label{eq:constituent_mp}
\end{equation}
Because $\mathcal S$ is the Cartesian product of the local signal-state sets and the candidate discharge of each real movement depends only on its own intersection state, the maximization in~\eqref{eq:constituent_mp} decomposes across intersections. Thus, the network expression in~\eqref{eq:constituent_pressure_matrix} is an analytical representation of decentralized MP phase selection rather than a centralized signal optimization.

This service-weighted formulation follows the finite-storage backpressure construction in which candidate signal states are evaluated using the discharge that can be realized under the prevailing traffic condition \citep{li2019position}. When $\bar d_a(\bm s,t)=c_as_a$, it reduces to the conventional saturation-flow MP objective \citep{varaiya2013max}. Importantly, the constituent-specific information enters through $w_a^h(t)$, whereas the phase-selection structure is shared. This separation allows Com-MP to combine heterogeneous MP preferences without requiring the constituents to use the same traffic-state definition.

\subsection{Heterogeneous sensing and representative controllers}
\label{sec:representative_controllers}

The constituent state $w_a^h(t)$ may be obtained from direct measurements or constructed from the sensing information available to the corresponding controller. The representative implementation uses fixed-detector information for a queue-based constituent and connected-vehicle (CV) observations for a travel-time-based constituent. These two sensing channels support complementary traffic-state representations while retaining the common MP structure in Section~\ref{sec:generalized_mp}.

\paragraph{Detector-supported queue-based MP (Q-MP) \citep{varaiya2013max}.}
For a real movement $a=(i,j)\in\mathcal M_{\mathcal N}$, Q-MP uses the length-normalized vehicle count
\begin{equation}
    w_a^Q(t)
    =
    \frac{z_a(t)}{L_i}.
    \label{eq:qmp_state}
\end{equation}
The normalization places vehicle accumulation on a density-like scale and avoids assigning larger pressure merely because an approach has greater physical storage. In implementation, the same state is constructed from the detector-supported information available for the movement; sensing degradation changes the available count information but does not change the Q-MP state definition in~\eqref{eq:qmp_state}. When usable detector-supported information is unavailable, the composition mechanism follows the sensing-availability rule specified in Section~\ref{sec:action}. Because real movements have finite storage, $w_a^Q(t)$ is uniformly bounded. Section~\ref{sec:qcv_stability} further verifies that this same Q-MP state admits the matched Lyapunov representation required by Assumption~\ref{ass:constituent_admissibility}.

\paragraph{CV-supported travel-time-based MP (CV-MP) \citep{tan2026cvmp}.}
Let $\mathcal J_a^{CV}(t)$ denote the set of connected vehicles observed for movement $a$ at time $t$. For $v\in\mathcal J_a^{CV}(t)$, let $\tau_v(t)$ denote the elapsed travel time of vehicle $v$ on its current link, and let $\bar\tau_a>0$ denote the corresponding free-flow travel time used to normalize this elapsed time. The CV-MP state is
\begin{equation}
    w_a^{CV}(t)
    =
    \sum_{v\in\mathcal J_a^{CV}(t)}
    \frac{\tau_v(t)}{\bar\tau_a}.
    \label{eq:cv_travel_time_state}
\end{equation}
Because each vehicle's contribution increases while it remains on the link, the resulting pressure gives increasing importance to movements that have accumulated temporal cost. Its representativeness also depends on the availability of CV observations. Section~\ref{sec:qcv_stability} establishes the boundedness and matched Lyapunov property required to place this state within the general stability framework.

The two representative constituents therefore differ in both operational emphasis and sensing dependence. Q-MP emphasizes vehicle accumulation using detector-supported information, whereas CV-MP emphasizes accumulated travel time using vehicle-level observations. Section~\ref{sec:com_mp} develops their adaptive composition while preserving the common MP phase-selection mechanism.

\section{Composite Max-Pressure Control and Stability Analysis}
\label{sec:com_mp}

Building on the constituent MP controllers introduced in Section~\ref{sec:problem_setting}, Com-MP adapts their relative influence according to the prevailing traffic and sensing conditions while retaining MP as the lower-level phase-selection mechanism. As illustrated in Fig.~\ref{FIG:framework}, each constituent controller first evaluates the candidate phases using its own pressure formulation. These phase pressures are normalized to represent comparable phase preferences and are then combined using intersection-specific composition weights. The composition-selection policy is developed in Section~\ref{sec:rl_composition}; this section defines the resulting controller and establishes the structural properties that preserve its network-level queue stability.

\begin{figure}
    \centering
    \includegraphics[width=\textwidth]{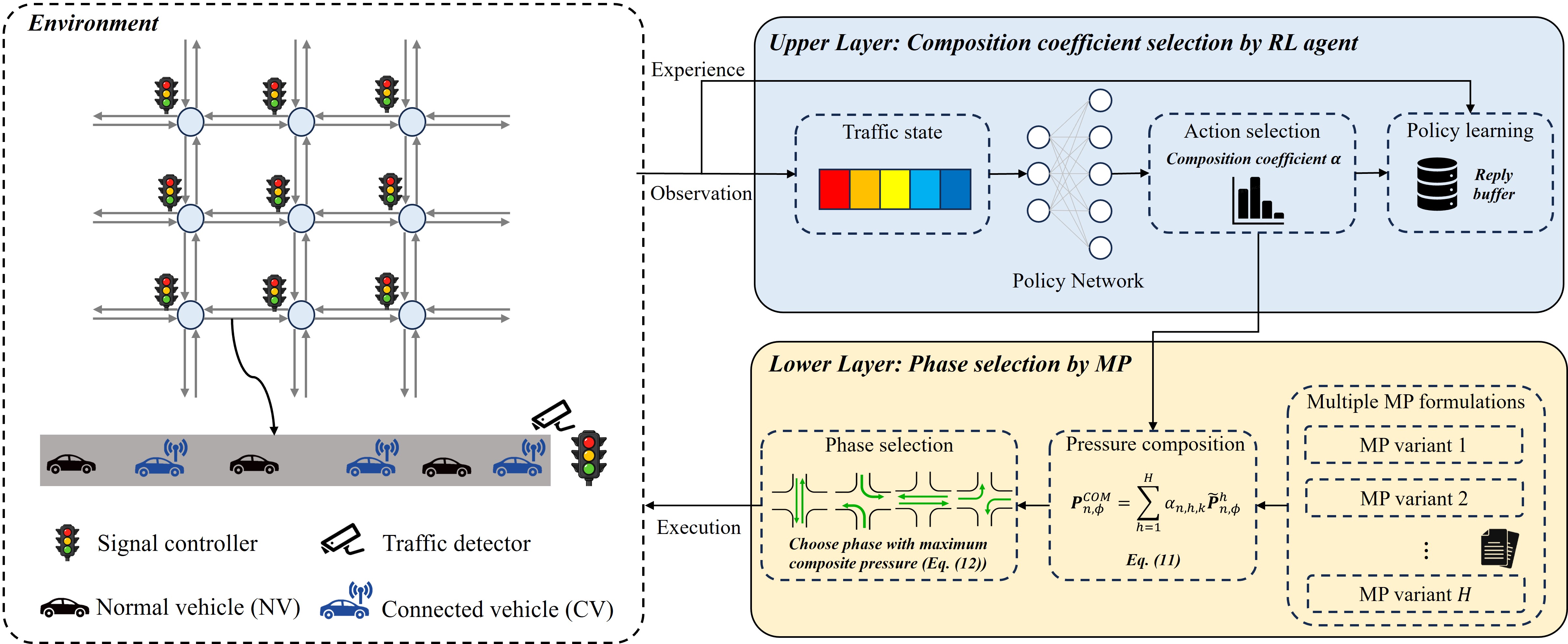}
    \caption{Framework of the proposed Com-MP controller. The upper layer adapts the composition of normalized constituent phase pressures, while the lower layer retains decentralized MP phase selection.}
    \label{FIG:framework}
\end{figure}

\subsection{Adaptive composite MP}
\label{sec:composite_formulation}

At each signal decision epoch $t_k$ defined in Section~\ref{sec:network_model}, Com-MP adapts only the relative influence of the constituent MP controllers, while the MP principle remains responsible for the actual phase decision. For intersection $n\in\mathcal N$, let $\bm o_{n,k}$ denote the local traffic and sensing observation available at $t_k$. The composition mechanism selects
\[
    \bm\alpha_{n,k}
    =
    [\alpha_{n,1,k},\ldots,\alpha_{n,H,k}]^\top
    \in\Delta_H,
    \qquad
    \Delta_H
    \triangleq
    \{\bm\alpha\in\mathbb R_+^H:\bm 1^\top\bm\alpha=1\}.
\]
Here, $\Delta_H$ is the $H$-dimensional probability simplex, $\mathbb R_+^H$ is the nonnegative orthant, $\bm 1$ is the all-ones vector, $h\in\mathcal H=\{1,\ldots,H\}$ indexes the constituent MP controllers, and $\alpha_{n,h,k}$ specifies their relative influence at intersection $n$. The selected composition is held fixed over $I_k$.

The constituent controllers may use traffic states with different physical meanings and numerical scales. Directly combining their raw pressures would therefore make the composition coefficients depend not only on control preference but also on the magnitude of the underlying pressure representation. Com-MP instead combines their relative preferences among the candidate phases. For constituent $h$ and candidate phase $\phi\in\mathcal P_n$, define the local phase pressure
\begin{equation}
    P_{n,\phi}^{h}(t_k)
    \triangleq
    \sum_{a\in\mathcal M_n}
    p_a^h(t_k)\,
    \bar d_a(\bm s_n^{\phi},t_k).
    \label{eq:constituent_phase_pressure}
\end{equation}
Let $P_{n,\min}^{h}(t_k)=\min_{\phi\in\mathcal P_n}P_{n,\phi}^{h}(t_k)$ and $P_{n,\max}^{h}(t_k)=\max_{\phi\in\mathcal P_n}P_{n,\phi}^{h}(t_k)$, and define their range as $\Gamma_{n,h,k}=P_{n,\max}^{h}(t_k)-P_{n,\min}^{h}(t_k)$. The normalized phase preference is
\begin{equation}
    \widetilde P_{n,\phi}^{h}(t_k)
    =
    \begin{cases}
    \dfrac{
        P_{n,\phi}^{h}(t_k)-P_{n,\min}^{h}(t_k)
    }{
        \Gamma_{n,h,k}
    },
    & \Gamma_{n,h,k}>0,\\[3mm]
    0,
    & \Gamma_{n,h,k}=0.
    \end{cases}
    \label{eq:normalized_phase_pressure}
\end{equation}
For $\Gamma_{n,h,k}>0$, this transformation maps the constituent phase pressures to a common unit interval without changing their phase ordering. When $\Gamma_{n,h,k}=0$, constituent $h$ is indifferent among all candidate phases and therefore contributes no phase preference at that decision epoch.

Com-MP combines the normalized preferences according to
\begin{equation}
    P_{n,\phi}^{COM}(t_k;\bm\alpha_{n,k})
    =
    \sum_{h=1}^{H}
    \alpha_{n,h,k}\,
    \widetilde P_{n,\phi}^{h}(t_k).
    \label{eq:composite_phase_pressure}
\end{equation}
Thus, the composition weights determine which constituent control preferences are emphasized under the current local traffic and sensing conditions, while the lower-level decision remains a pressure maximization.

The phase decision must also satisfy the maximum-red requirement in~\eqref{eq:maximum_red_requirement}. Let $\mathcal P_n^{\mathrm{adm}}(t_k)\subseteq\mathcal P_n$ denote the nonempty set of phases that are operationally admissible at $t_k$. Intersection $n$ selects
\begin{equation}
    \phi_{n,k}^{*}
    \in
    \arg\max_{\phi\in\mathcal P_n^{\mathrm{adm}}(t_k)}
    P_{n,\phi}^{COM}(t_k;\bm\alpha_{n,k}),
    \label{eq:com_mp_control}
\end{equation}
and holds the selected phase over $I_k$, apart from the finite transition interval required when the phase changes. Hence, the upper layer modifies the composition of MP preferences rather than replacing the decentralized MP phase-selection rule.

\paragraph{Representative Q--CV composition.}
The experimental implementation combines the detector-supported Q-MP and CV-supported travel-time MP introduced in Section~\ref{sec:representative_controllers}. Setting $H=2$, let $\alpha_{n,k}\in[0,1]$ denote the relative contribution of CV-MP. The composite phase pressure becomes
\begin{equation}
    P_{n,\phi}^{COM}(t_k;\alpha_{n,k})
    =
    (1-\alpha_{n,k})\widetilde P_{n,\phi}^{Q}(t_k)
    +
    \alpha_{n,k}\widetilde P_{n,\phi}^{CV}(t_k),
    \label{eq:qcv_composite_pressure}
\end{equation}
and the corresponding phase decision is
\begin{equation}
    \phi_{n,k}^{*}
    \in
    \arg\max_{\phi\in\mathcal P_n^{\mathrm{adm}}(t_k)}
    \left[
        (1-\alpha_{n,k})\widetilde P_{n,\phi}^{Q}(t_k)
        +
        \alpha_{n,k}\widetilde P_{n,\phi}^{CV}(t_k)
    \right].
    \label{eq:qcv_com_mp_control}
\end{equation}
Thus, $\alpha_{n,k}=0$ and $\alpha_{n,k}=1$ recover the normalized Q-MP and CV-MP phase preferences, respectively, whenever the corresponding constituent has a nonzero phase-pressure range. If no usable CV state is available, we set $\alpha_{n,k}=0$ so that the controller relies on Q-MP.

\subsection{Stability analysis}
\label{sec:stability_analysis}

The main theoretical question is whether the MP stability property is preserved when the pressure representation is normalized, composed differently across intersections, and allowed to change over time. The analysis follows the finite-storage backpressure framework of \citet{li2019position} and the generalized MP formulation used by \citet{tan2026cvmp}. The argument proceeds in three steps. First, the implemented normalized composition is shown to be phase-equivalent to a convex combination of the raw constituent pressures. Second, this composition is related to a network-consistent analytical MP objective, with only a uniformly bounded decision loss. Finally, the resulting negative drift is combined with the bounded effect of composition updates to establish strong stability.

\subsubsection{Stability setting}
\label{sec:stability_conditions}

Because real movements have finite physical storage, their vehicle counts are uniformly bounded. The potentially unbounded components are the fictitious source queues representing demand waiting to enter the network. We nevertheless state stability for the total movement count.

\begin{definition}[Strong queue stability]
\label{def:strong_stability}
For an averaging horizon $T>0$, the traffic network is strongly stable if
\begin{equation}
    \limsup_{T\rightarrow\infty}
    \frac{1}{T}
    \int_0^T
    \mathbb E
    \left[
        \sum_{a\in\mathcal M}z_a(t)
    \right]dt
    <
    \infty.
    \label{eq:strong_stability}
\end{equation}
\end{definition}

For a Lyapunov process $V(t)$, let $\mathcal L_{\bm s}V(t)$ denote its conditional extended infinitesimal drift under candidate signal state $\bm s$. When the signal state is generated by the applied policy rather than fixed as a candidate, we write $\mathcal L V(t)$ for the corresponding conditional drift. The operator includes the continuous evolution between traffic-transition events together with jumps caused by arrivals and movement transitions.

\begin{lemma}[Sufficient drift condition for strong stability]
\label{lem:strong_stability_criterion}
Let $V(t)\ge0$ be a Lyapunov process whose extended drift is well defined and integrable. If there exist finite constants $K<\infty$ and $\epsilon>0$ such that
\begin{equation}
    \mathcal L V(t)
    \le
    K
    -
    \epsilon
    \sum_{a\in\mathcal M_{\mathcal F}}z_a(t),
    \label{eq:strong_stability_criterion}
\end{equation}
then the finite-storage traffic network is strongly stable in the sense of Definition~\ref{def:strong_stability}.
\end{lemma}

The result follows by integrating the drift inequality over time; real-movement counts are already bounded by finite storage. The complete argument is given in Appendix~\ref{app:proof_strong_stability_criterion}.

We use the following regularity conditions for the network dynamics.

\begin{assumption}[Traffic-process regularity]
\label{ass:network_regularity}
The phase sets $\{\mathcal P_n\}_{n\in\mathcal N}$ and downstream-movement sets $\{\mathcal D(a)\}_{a\in\mathcal M}$ are fixed. The turning-ratio process is piecewise constant on the decision intervals, with $\bm R(t)=\bm R_k$ for $t\in I_k$, where $\bm R_k$ is $\mathcal I_{t_k^-}$-measurable, nonnegative, row-substochastic, and has the structure defined in Section~\ref{sec:network_model}. For each real movement $a$, let $z_a^{\max}<\infty$ denote its finite physical storage bound, so that $0\le z_a(t)\le z_a^{\max}$. Traffic densities, movement flows, vehicle-transition rates, and exogenous arrival rates are bounded by their physical finite-link limits, and the exogenous arrivals satisfy $\mathbb E[\bm\lambda(t)\mid\mathcal I_{t^-}]=\bar{\bm\lambda}$. For every feasible candidate signal state, the conditional expected movement discharge $\bar d_a(\bm s,t)$ is bounded and has uniformly bounded variation over each decision interval of duration $\Delta$. For $a\in\mathcal M_n$, candidate discharge depends on $\bm s$ only through the local signal state $\bm s_n$; fictitious-source discharge is independent of the candidate real-intersection phases.
\end{assumption}

Assumption~\ref{ass:network_regularity} collects standard regularity conditions used in finite-storage backpressure and MP stability analyses \citep{li2019position,tan2026cvmp}. In particular, finite physical storage implies bounded real-movement counts and service rates, while the piecewise-constant turning matrix reflects the decision-epoch update used by the controller. The conditional-mean condition on exogenous arrivals does not require independent or deterministic arrivals; it provides the stationary mean demand with respect to which the long-run stability region is defined. Combining this condition with the conservation relation in~\eqref{eq:network_count_dynamics} and the candidate expected discharge in~\eqref{eq:candidate_expected_discharge}, the conditional count drift under candidate signal state $\bm s$ is
\begin{equation}
    \mathcal L_{\bm s}\bm z(t)
    =
    \bar{\bm\lambda}
    -
    [\bm I-\bm R^\top(t)]\bar{\bm d}(\bm s,t).
    \label{eq:candidate_count_drift}
\end{equation}
For an absolutely continuous representation,~\eqref{eq:candidate_count_drift} is the conditional rate form of~\eqref{eq:network_count_dynamics}. Under an event-based representation, it gives the corresponding first-order generator identity. Bounded jump corrections enter the Lyapunov drift separately.

The constituent controllers may represent congestion differently. Rather than requiring a common state definition, the stability analysis only requires each constituent state to admit a compatible Lyapunov representation.

\begin{assumption}[Admissible constituent MP states]
\label{ass:constituent_admissibility}
For each constituent controller $h\in\mathcal H$ and real movement $a$, let $w_a^h(t)$ be the state entering its MP pressure. Associated with this state is a nonnegative movement-level Lyapunov term $G_a^h(t)$. There exist finite constants $W_a^h$, $G_a^{h,\max}$, and $K_a^h$, uniform over time, such that
\begin{equation}
    0\le w_a^h(t)\le W_a^h,
    \qquad
    0\le G_a^h(t)\le G_a^{h,\max},
    \label{eq:constituent_admissibility}
\end{equation}
and, for every feasible candidate signal state $\bm s$,
\begin{equation}
    \mathcal L_{\bm s}G_a^h(t)
    \le
    w_a^h(t)\,
    \mathcal L_{\bm s}z_a(t)
    +
    K_a^h.
    \label{eq:constituent_potential_drift}
\end{equation}
For fictitious source movements, all constituent controllers use the physical source queue $w_a^h(t)=z_a(t)$ and the corresponding Lyapunov term $z_a^2(t)/2$. Here, $W_a^h$ and $G_a^{h,\max}$ are uniform bounds on the constituent state and its Lyapunov term, respectively, while $K_a^h$ bounds the residual term in the matched drift relation.
\end{assumption}

Assumption~\ref{ass:constituent_admissibility} defines the class of constituent state representations that can be incorporated into the composite stability analysis. It does not assume that an individual constituent controller is itself stabilizing, nor does it require the constituents to use the same traffic state. Instead, it requires the real-movement state contribution to remain bounded under finite storage and to admit a Lyapunov potential whose drift is matched to the state-weighted vehicle-count dynamics up to a bounded remainder. This is the same type of state--drift compatibility underlying generalized backpressure constructions \citep{li2019position,tan2026cvmp}. This condition therefore provides the interface between a constituent MP formulation and the composite stability analysis. Section~\ref{sec:qcv_stability} verifies these properties explicitly for the Q-MP and CV-MP constituents used in the experiments from finite storage and explicit service conditions.

The stability guarantee is stated relative to demands that admit a strict stability margin under the same network dynamics and operational signal constraints.

\begin{definition}[Strictly stabilizable demand region]
\label{def:strict_stabilizable_region}
A mean demand vector $\bar{\bm\lambda}$ is called \emph{strictly stabilizable} if there exist an admissible reference policy $\sigma^{\mathrm{ref}}\in\Sigma_{\mathrm{adm}}$, a reference movement state $\bm w^{\mathrm{ref}}(t)$, and a matched Lyapunov functional $V^{\mathrm{ref}}$. The reference state uses the physical source queues on $\mathcal M_{\mathcal F}$ and is uniformly bounded on $\mathcal M_{\mathcal N}$. Let $\Psi^{\mathrm{ref}}(\bm s,t)$ denote the pressure-service objective induced by $\bm w^{\mathrm{ref}}(t)$ through the same differential-pressure construction as in~\eqref{eq:movement_pressure}--\eqref{eq:constituent_pressure_matrix}, and let $\bm s^{\mathrm{ref}}(t)$ denote the signal state selected by $\sigma^{\mathrm{ref}}$. The Lyapunov functional is matched to the reference state if there exists a residual $\xi^{\mathrm{ref}}(t)$ and a finite constant $K_{\mathrm{ref}}<\infty$ such that
\begin{equation}
    \mathcal L_{\sigma^{\mathrm{ref}}}V^{\mathrm{ref}}(t)
    =
    (\bm w^{\mathrm{ref}}(t))^\top\bar{\bm\lambda}
    -
    \Psi^{\mathrm{ref}}(\bm s^{\mathrm{ref}}(t),t)
    +
    \xi^{\mathrm{ref}}(t),
    \qquad
    |\xi^{\mathrm{ref}}(t)|\le K_{\mathrm{ref}}.
    \label{eq:reference_matched_drift}
\end{equation}
In addition, there exist finite $K_0<\infty$ and $\epsilon>0$ such that
\[
    \mathcal L_{\sigma^{\mathrm{ref}}}V^{\mathrm{ref}}(t)
    \le
    K_0
    -
    \epsilon
    \sum_{a\in\mathcal M_{\mathcal F}}z_a(t)
\]
uniformly over admissible traffic states. The set of all such mean demand vectors is
\begin{equation}
    \bm\Lambda^\circ
    \triangleq
    \left\{
        \bar{\bm\lambda}:
        \bar{\bm\lambda}
        \text{ is strictly stabilizable under }
        \Sigma_{\mathrm{adm}}
    \right\}.
    \label{eq:strict_stabilizable_region}
\end{equation}
Its closure $\bm\Lambda\triangleq\operatorname{cl}(\bm\Lambda^\circ)$, where $\operatorname{cl}(\cdot)$ denotes set closure, is referred to as the admissible demand region.
\end{definition}

The bounded residual in~\eqref{eq:reference_matched_drift} makes explicit that the reference Lyapunov drift has the same arrival-minus-pressure-service structure as the generalized MP construction, up to a uniformly bounded term. The reference policy is therefore a stability witness establishing a strict service margin under the admissible finite-storage network dynamics. It need not be one of the constituent MP controllers or the proposed Com-MP policy. The analysis below shows that adaptive composition preserves this strict margin up to bounded terms.

\subsubsection{Structural properties of adaptive composition}

The implemented controller combines normalized phase pressures, whereas the generalized MP stability framework is expressed in terms of the raw constituent states. The first result connects these two representations.

\begin{proposition}[Phase-equivalent representation of normalized composition]
\label{prop:normalized_composition_equivalence}
Fix intersection $n$ and decision epoch $t_k$. There exists an effective raw-pressure composition vector $\bm\theta_{n,k}\in\Delta_H$ such that, for every nonempty candidate subset $\mathcal P_n'\subseteq\mathcal P_n$,
\begin{equation}
    \arg\max_{\phi\in\mathcal P_n'}
    \sum_{h=1}^{H}
    \alpha_{n,h,k}\widetilde P_{n,\phi}^{h}(t_k)
    =
    \arg\max_{\phi\in\mathcal P_n'}
    \sum_{h=1}^{H}
    \theta_{n,h,k}P_{n,\phi}^{h}(t_k).
    \label{eq:normalized_raw_equivalence}
\end{equation}
\end{proposition}

Define $S_{n,k}\triangleq\sum_{g:\Gamma_{n,g,k}>0}\alpha_{n,g,k}/\Gamma_{n,g,k}$, where $g$ is a dummy constituent index. If $S_{n,k}>0$, the equivalent coefficients are
\begin{equation}
    \theta_{n,h,k}
    =
    \begin{cases}
    \dfrac{\alpha_{n,h,k}/\Gamma_{n,h,k}}{S_{n,k}},
    & \Gamma_{n,h,k}>0,\\[3mm]
    0,
    & \Gamma_{n,h,k}=0.
    \end{cases}
    \label{eq:effective_composition}
\end{equation}
If $S_{n,k}=0$, all positively weighted constituents are phase indifferent and we set $\bm\theta_{n,k}=\bm\alpha_{n,k}$.

Proposition~\ref{prop:normalized_composition_equivalence} shows that normalization changes the relative scaling of the constituent pressures but does not introduce a different phase-selection principle. At every intersection and decision epoch, the implemented normalized composition is phase-equivalent to a convex combination of the raw constituent MP pressures. Importantly, the equivalence holds for any nonempty candidate subset and therefore remains valid when the maximum-red rule restricts the admissible phases. The proof is given in Appendix~\ref{app:proof_normalized_composition}.

Because the effective composition may differ across intersections, a single global composition vector generally does not exist. For the stability analysis, assign each real movement the composition of its owning intersection. Let $\nu(a)$ denote the owner of real movement $a$ and define, for $t\in I_k$,
\begin{equation}
    \bar w_a^{k}(t)
    \triangleq
    \sum_{h=1}^{H}
    \theta_{\nu(a),h,k}w_a^h(t),
    \qquad
    \bar G_a^{k}(t)
    \triangleq
    \sum_{h=1}^{H}
    \theta_{\nu(a),h,k}G_a^h(t).
    \label{eq:analytical_composite_state}
\end{equation}
For $a\in\mathcal M_{\mathcal F}$, set $\bar w_a^{k}(t)=z_a(t)$. Let $\bar{\bm w}^{k}(t)$ collect the analytical movement states $\{\bar w_a^k(t)\}_{a\in\mathcal M}$. The corresponding network-consistent analytical movement pressure is
\begin{equation}
    \bar p_a^{k}(t)
    =
    \bar w_a^{k}(t)
    -
    \sum_{b\in\mathcal D(a)}
    r_b(t)\bar w_b^{k}(t).
    \label{eq:analytical_pressure}
\end{equation}

\begin{lemma}[Closure of the constituent-state class under composition]
\label{lem:composite_admissibility}
Under Assumption~\ref{ass:constituent_admissibility}, every owner-based convex composition in~\eqref{eq:analytical_composite_state} remains in the same admissible state class. In particular, for every real movement $a$,
\begin{equation}
    0\le\bar w_a^{k}(t)\le\bar W_a,
    \qquad
    0\le\bar G_a^{k}(t)\le\bar G_a,
    \label{eq:composite_admissibility}
\end{equation}
for finite constants $\bar W_a=\max_h W_a^h$ and $\bar G_a=\max_hG_a^{h,\max}$, and
\begin{equation}
    \mathcal L_{\bm s}\bar G_a^{k}(t)
    \le
    \bar w_a^{k}(t)\,
    \mathcal L_{\bm s}z_a(t)
    +
    \bar K_a,
    \qquad
    \bar K_a=\max_hK_a^h.
    \label{eq:composite_potential_drift}
\end{equation}
\end{lemma}

Thus, convex composition does not take the analytical state outside the class for which the constituent drift relation is available. The proof follows directly from the simplex property of $\bm\theta_{n,k}$ and is given in Appendix~\ref{app:proof_composite_admissibility}.

To connect this composite state to network stability, define, for a fixed analytical composition over $I_k$,
\begin{equation}
    V_k(t)
    =
    \sum_{a\in\mathcal M_{\mathcal N}}
    \bar G_a^{k}(t)
    +
    \frac{\kappa}{2}
    \sum_{a\in\mathcal M_{\mathcal N}}z_a^2(t)
    +
    \frac{1}{2}
    \sum_{a\in\mathcal M_{\mathcal F}}z_a^2(t),
    \qquad
    \kappa>0,
    \label{eq:composite_lyapunov}
\end{equation}
where $\kappa$ is an arbitrary fixed positive coefficient on the bounded real-movement quadratic term, and define the corresponding analytical MP objective
\begin{equation}
    \Psi_k(\bm s,t)
    \triangleq
    \sum_{a\in\mathcal M}
    \bar p_a^{k}(t)\,
    \bar d_a(\bm s,t).
    \label{eq:analytical_service_objective}
\end{equation}
The quadratic real-movement term in~\eqref{eq:composite_lyapunov} remains bounded under finite storage and therefore contributes only a finite drift remainder.

\begin{lemma}[Composite Lyapunov drift]
\label{lem:service_weighted_drift}
Under Assumptions~\ref{ass:network_regularity}--\ref{ass:constituent_admissibility}, there exists a finite constant $K_V$, uniform over the induced analytical compositions, such that
\begin{equation}
    \mathcal L_{\bm s}V_k(t)
    \le
    (\bar{\bm w}^{k}(t))^\top\bar{\bm\lambda}
    -
    \Psi_k(\bm s,t)
    +
    K_V
    \label{eq:service_weighted_drift}
\end{equation}
for every feasible candidate signal state $\bm s$.
\end{lemma}

Lemma~\ref{lem:service_weighted_drift} gives the usual arrival-minus-pressure drift structure for the composite state. Its derivation is provided in Appendix~\ref{app:proof_service_drift}. What remains is to relate the decentralized phase decisions actually implemented by Com-MP to the network-consistent objective $\Psi_k$.

\begin{proposition}[Bounded-error network MP property]
\label{prop:bounded_error_mp}
Let $\bm s_k^*=(\bm s_n^{\phi_{n,k}^*})_{n\in\mathcal N}\in\mathcal S$ denote the network signal state induced by the local phases selected through~\eqref{eq:com_mp_control} at $t_k$. There exists a finite constant $B$, independent of the fictitious-source queue magnitudes, the selected composition profile, and the current admissible movement-transition matrix $\bm R_k$, such that
\begin{equation}
    \Psi_k(\bm s_k^*,t)
    \ge
    \max_{\bm s\in\mathcal S}
    \Psi_k(\bm s,t)
    -
    B,
    \qquad
    t\in I_k.
    \label{eq:bounded_error_mp}
\end{equation}
\end{proposition}

The bounded gap has three sources. First, neighboring intersections may use different composition vectors, whereas $\Psi_k$ uses one network-consistent analytical state for each movement. Second, the maximum-red requirement can restrict the local phase choice away from the unconstrained MP maximizer. Third, the selected phase is held over the finite interval $I_k$, during which the real-movement part of the analytical objective can vary. Under Assumptions~\ref{ass:network_regularity}--\ref{ass:constituent_admissibility}, all three effects are uniformly bounded; their explicit bounds are derived in Appendix~\ref{app:proof_bounded_error_mp}.

The key property is that $B$ is independent of the unbounded source-queue magnitudes. Adaptive and intersection-specific pressure composition therefore perturbs the network MP objective only by a bounded amount, preserving the strict queue-dependent drift margin.

\subsubsection{Strong-stability guarantee}
\label{sec:stability_proof_main}

The preceding results establish the two ingredients needed for the network-level guarantee: negative drift within each composition interval and a bounded effect when the analytical composition changes.

\begin{lemma}[Negative drift within a composition interval]
\label{lem:fixed_action_drift}
For every $\bar{\bm\lambda}\in\bm\Lambda^\circ$, there exist finite constants $K<\infty$ and $\epsilon>0$, uniform over the selected composition profile, such that
\begin{equation}
    \mathcal L_{\bm s_k^*}V_k(t)
    \le
    K
    -
    \epsilon
    \sum_{a\in\mathcal M_{\mathcal F}}z_a(t),
    \qquad
    t\in I_k.
    \label{eq:fixed_action_drift}
\end{equation}
\end{lemma}

The strict margin available to the reference policy in Definition~\ref{def:strict_stabilizable_region} is preserved because the reference and composite constructions use the same physical source queues, whereas their real-movement contributions differ only by bounded terms. Proposition~\ref{prop:bounded_error_mp} then transfers this margin to the implemented Com-MP action. The complete argument is given in Appendix~\ref{app:proof_fixed_action_drift}.

Composition coefficients may change at every decision epoch. The next lemma establishes a uniform bound on the Lyapunov-function change induced by each composition update.

\begin{lemma}[Bounded analytical-composition update]
\label{lem:bounded_composition_jump}
There exists a finite constant $K_{\mathrm{sw}}$ such that
\begin{equation}
    \left|
        V_k(t_k^+)-V_{k-1}(t_k^-)
    \right|
    \le
    K_{\mathrm{sw}}
    \label{eq:bounded_composition_jump}
\end{equation}
at every decision epoch, where $t_k^-$ and $t_k^+$ denote the instants immediately before and after the analytical composition update at $t_k$, respectively.
\end{lemma}

At an update instant, the physical traffic state is unchanged and only the analytical coefficients are replaced. Since both the old and new coefficient vectors lie in the simplex and all real-movement constituent Lyapunov terms are bounded, the resulting jump is uniformly bounded. The proof is given in Appendix~\ref{app:proof_composition_jump}.

\begin{theorem}[Strong stability of adaptive Com-MP]
\label{thm:com_mp_stability}
Suppose Assumptions~\ref{ass:network_regularity}--\ref{ass:constituent_admissibility} hold. For every strictly stabilizable mean demand vector $\bar{\bm\lambda}\in\bm\Lambda^\circ$, let each intersection select an arbitrary nonanticipative composition vector $\bm\alpha_{n,k}\in\Delta_H$ at every decision epoch, form the normalized phase preferences in~\eqref{eq:normalized_phase_pressure}, and apply the phase rule in~\eqref{eq:com_mp_control} subject to the maximum-red requirement~\eqref{eq:maximum_red_requirement}. Then adaptive intersection-specific Com-MP strongly stabilizes the network in the sense of Definition~\ref{def:strong_stability}.
\end{theorem}

Theorem~\ref{thm:com_mp_stability} shows that pressure composition may vary both spatially and temporally without reducing the strictly stabilizable region characterized in Definition~\ref{def:strict_stabilizable_region}. In particular, the guarantee does not depend on the specific algorithm used to generate the composition coefficients: they may be produced by a learned policy, a heuristic rule, or another nonanticipative mechanism. Learning can therefore be used to improve how constituent MP preferences are combined without becoming part of the underlying stability condition.

\begin{proof}[Proof sketch]
Proposition~\ref{prop:normalized_composition_equivalence} converts the implemented normalized phase preference into an equivalent raw-pressure composition, and Lemma~\ref{lem:composite_admissibility} shows that the resulting analytical state remains in the admissible constituent class. Proposition~\ref{prop:bounded_error_mp} bounds the loss of the implemented decentralized decision relative to the corresponding network MP objective by a constant independent of the source queues. Lemma~\ref{lem:fixed_action_drift} therefore gives negative drift throughout every decision interval.

The composition may change between intervals, but Lemma~\ref{lem:bounded_composition_jump} bounds each resulting Lyapunov jump by $K_{\mathrm{sw}}$. For any integer $m\ge1$, let $t_m=m\Delta$ denote the end of the first $m$ complete decision intervals. Summing over these intervals gives
\[
    \mathbb E[V_m(t_m^+)]
    -
    \mathbb E[V_0(t_0^+)]
    \le
    m(K\Delta+K_{\mathrm{sw}})
    -
    \epsilon
    \int_0^{t_m}
    \mathbb E
    \left[
        \sum_{a\in\mathcal M_{\mathcal F}}z_a(t)
    \right]dt.
\]
Since $V_m(t_m^+)\ge0$ and $t_m=m\Delta$, the long-run average fictitious-source queue is finite. Real-movement counts are uniformly bounded by finite storage, which yields Definition~\ref{def:strong_stability}. The complete telescoping argument is given in Appendix~\ref{app:completion_com_mp_stability}.
\end{proof}

A finite yellow or intergreen interval contributes only a uniformly bounded amount to the drift over one decision interval because real-movement states and movement discharges are bounded. This contribution is absorbed into the finite drift constants above.

\subsection{Q--CV Com-MP specialization}
\label{sec:qcv_stability}

Theorem~\ref{thm:com_mp_stability} applies to any collection of constituent MP states satisfying Assumption~\ref{ass:constituent_admissibility}. We now verify that the queue-based and CV-supported travel-time states used in the representative Q--CV formulation satisfy these conditions under an explicit service assumption.

For Q-MP, the state in~\eqref{eq:qmp_state} and its movement-level Lyapunov term are
\[
    w_a^Q(t)=\frac{z_a(t)}{L_i},
    \qquad
    G_a^Q(t)=\frac{z_a^2(t)}{2L_i},
    \qquad
    a=(i,j)\in\mathcal M_{\mathcal N}.
\]
Finite storage bounds both $w_a^Q(t)$ and $G_a^Q(t)$, and their continuous and event-level dynamics satisfy the matched drift relation in~\eqref{eq:constituent_potential_drift} up to a finite remainder.

For CV-MP, the state in~\eqref{eq:cv_travel_time_state} and the associated movement-level Lyapunov term are
\[
    w_a^{CV}(t)
    =
    \sum_{v\in\mathcal J_a^{CV}(t)}
    \frac{\tau_v(t)}{\bar\tau_a},
    \qquad
    G_a^{CV}(t)
    =
    z_a(t)w_a^{CV}(t).
\]
The event-level evolution of this potential also satisfies~\eqref{eq:constituent_potential_drift}. Its boundedness additionally requires that vehicles occupying a finite movement receive recurrent effective service.

\begin{assumption}[Recoverable movement service]
\label{ass:recoverable_cv}
Traffic on each real movement follows first-in-first-out (FIFO) order, and the non-control running component of travel time on a finite real link is uniformly bounded. For each real movement $a$, let $\phi(a)$ denote a phase that serves it, and call a decision interval in which $\phi(a)$ is selected a serving activation. If $a$ remains continuously occupied, there exist an integer $m_a<\infty$ and a constant $\underline d_a>0$ such that the cumulative realized discharge of $a$ over any $m_a$ consecutive serving activations is at least $\underline d_a$.
\end{assumption}

The condition is expressed in terms of realized discharge and therefore incorporates downstream receiving restrictions and spillback. Together with the maximum-red requirement, it ensures recurrent positive service within a bounded time horizon.

\begin{proposition}[Admissibility of the Q--CV constituents]
\label{prop:qcv_admissibility}
Suppose Assumptions~\ref{ass:network_regularity} and~\ref{ass:recoverable_cv} hold. Under the maximum-red requirement~\eqref{eq:maximum_red_requirement}, every real movement $a$ admits a finite constant $\tau_a^{\max}$ such that the elapsed link travel time of every CV on that movement satisfies
\[
    \tau_v(t)\le\tau_a^{\max}.
\]
Consequently, the Q-MP and CV-MP states and their matched Lyapunov terms are uniformly bounded and satisfy Assumption~\ref{ass:constituent_admissibility}.
\end{proposition}

For Q-MP, the required bounds follow directly from finite movement storage. For CV-MP, finite storage and FIFO bound the amount of traffic that can remain ahead of a vehicle, while the maximum-red requirement and Assumption~\ref{ass:recoverable_cv} guarantee a finite amount of positive service over successive serving activations. A CV therefore cannot remain on a real movement indefinitely, which bounds its travel-time contribution and hence the CV-MP state. The event-level matched-potential calculations and the complete bound are given in Appendix~\ref{app:bounded_cv_travel_time}.

\begin{corollary}[Stability of Q--CV Com-MP]
\label{cor:q_cv_com_mp}
Suppose Assumptions~\ref{ass:network_regularity} and~\ref{ass:recoverable_cv} hold. Under the maximum-red requirement~\eqref{eq:maximum_red_requirement}, every nonanticipative intersection-wise Q--CV composition sequence generated through~\eqref{eq:qcv_composite_pressure} and~\eqref{eq:qcv_com_mp_control} strongly stabilizes the network for every $\bar{\bm\lambda}\in\bm\Lambda^\circ$.
\end{corollary}

\subsection{Trustworthy Learning through Stability Preservation}
\label{sec:learning_implications}

The preceding results provide a mathematical basis for trustworthy learning-enabled control. The stability analysis is uniform over the admissible composition profile and does not require the composition coefficients to be generated by a particular learning algorithm. Learning can therefore adapt the relative influence of the constituent MP preferences for performance, while network-level queue stability remains anchored in the structured Com-MP mechanism.

\begin{proposition}[Stability under learned composition]
\label{prop:policy_independent_learning}
Suppose the conditions of Theorem~\ref{thm:com_mp_stability} hold. Any nonanticipative learned policy that selects $\bm\alpha_{n,k}\in\Delta_H$ at each decision epoch preserves strong stability of the resulting Com-MP system for every $\bar{\bm\lambda}\in\bm\Lambda^\circ$.
\end{proposition}

\begin{proof}
Every realized learned policy induces an admissible nonanticipative composition sequence. The result therefore follows directly from Theorem~\ref{thm:com_mp_stability}.
\end{proof}

Proposition~\ref{prop:policy_independent_learning} makes the stability guarantee independent of how the composition policy is parameterized or trained. Retraining or replacing the learned policy may change how Com-MP adapts to local traffic and sensing conditions, but it does not alter the stability conclusion as long as the realized compositions remain admissible and nonanticipative. This separation gives one concrete dimension of trustworthiness: learning is used to improve context-dependent performance within a controller family whose network-level queue stability is certified by the underlying control structure. The result concerns closed-loop stability rather than convergence or optimality of the RL algorithm itself.

For the representative Q--CV implementation, Corollary~\ref{cor:q_cv_com_mp} further verifies that the two sensing-dependent constituent states satisfy the required conditions. The remaining task is therefore to learn how the admissible composition should change with local traffic and sensing conditions, which is addressed in Section~\ref{sec:rl_composition}.

\section{Theory-Guided RL for Adaptive Composition}
\label{sec:rl_composition}

Section~\ref{sec:com_mp} establishes that the lower-level Com-MP controller remains strongly stable under arbitrary nonanticipative, intersection-specific composition sequences satisfying the stated conditions, while Proposition~\ref{prop:policy_independent_learning} makes this guarantee explicit for learned composition policies. The remaining problem is therefore one of performance adaptation: how should each intersection select a useful composition under the current traffic and sensing context? We address this problem using multi-agent reinforcement learning (MARL), while keeping the actual signal-phase decision within the Com-MP rule. In the representative Q--CV implementation, the learned policy selects only the scalar coefficient $\alpha_{n,k}$; the lower-level controller then combines the normalized Q-MP and CV-MP phase pressures through~\eqref{eq:qcv_composite_pressure} and determines the signal phase through~\eqref{eq:qcv_com_mp_control}. Thus, learning adapts the relative influence of the constituent MP preferences within the stability-preserving control structure rather than replacing the MP phase-selection mechanism.

\subsection{Shared decentralized composition policy}
\label{sec:rl_architecture}

Each real signalized intersection $n\in\mathcal N$ is treated as an agent that observes local traffic and sensing information and independently selects its composition coefficient. All intersections share a stochastic actor policy $\pi_\psi$, parameterized by $\psi$, so that
\begin{equation}
    \alpha_{n,k}
    \sim
    \pi_{\psi}
    \left(
        \cdot
        \mid
        \bm o_{n,k}
    \right),
    \qquad
    n\in\mathcal N,
    \label{eq:local_actor}
\end{equation}
where $\bm o_{n,k}$ is the local observation at decision epoch $t_k$. Parameter sharing allows traffic experience from different intersections to train a common composition rule, while conditioning on local observations preserves intersection-specific decisions. The shared policy therefore provides a common local composition rule that can be applied across heterogeneous intersection structures. The overall learning architecture is illustrated in Fig.~\ref{FIG:ctde}.

\begin{figure}
    \centering
    \includegraphics[width=\textwidth]{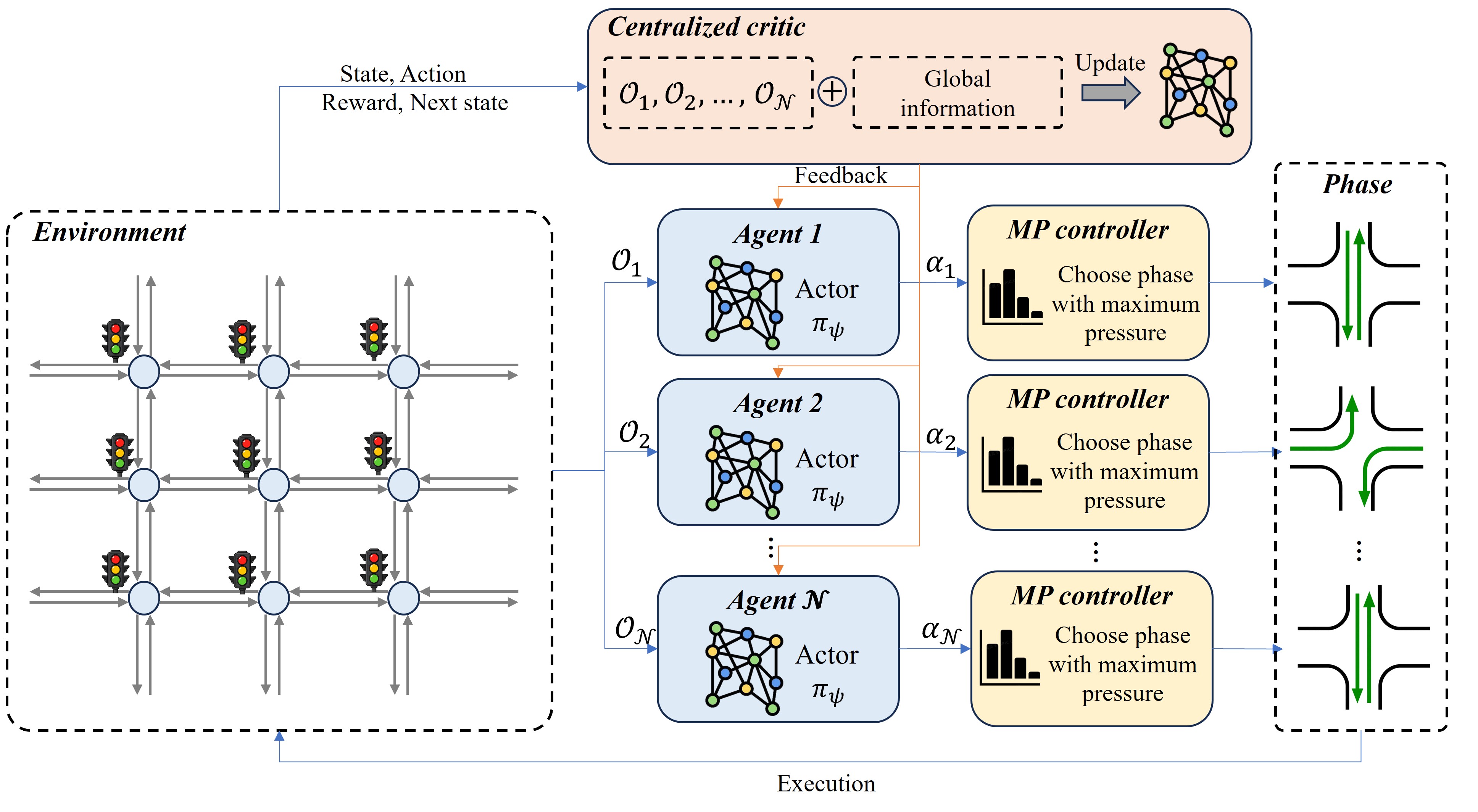}
    \caption{Centralized-training-and-decentralized-execution architecture for learning the intersection-wise Com-MP composition policy.}
    \label{FIG:ctde}
\end{figure}

Because each agent observes only local information while traffic interactions remain network-coupled, the learning problem is partially observable and multi-agent. We adopt multi-agent proximal policy optimization (MAPPO) under a centralized-training-and-decentralized-execution (CTDE) architecture. During training, a centralized critic augments each local observation with compact network-level context formed from aggregate intersection observations, which improves value estimation for locally selected compositions \citep{yu2022surprising}. During execution, the critic and network-level context are removed; each intersection uses only the shared actor in~\eqref{eq:local_actor} and its own local observation. This preserves the decentralized execution structure of Com-MP.

\subsection{Local composition decision}
\label{sec:composition_decision_rl}

Given the shared decentralized policy above, each intersection determines its composition from information describing the current traffic condition, constituent pressure preferences, and their sensing support. We therefore specify the local observation supplied to the actor, the composition action through which the actor adjusts the lower-level Com-MP controller, and the reward used to evaluate the resulting traffic effect. Together, these elements define the local performance-adaptation problem while leaving signal-phase selection to the MP mechanism.

\subsubsection{Local observation}
\label{sec:local_observation}

The observation is designed to describe three factors relevant to the relative usefulness of the constituent controllers: their current pressure preferences, the local traffic condition, and the sensing support available for constructing their states. For intersection $n$ at decision epoch $t_k$, we use
\begin{equation}
    \bm o_{n,k}
    =
    \left[
        \left\{
            \{\widetilde{\varpi}_{n,a,k}^{\,h}\}_{a\in\mathcal M_n}
        \right\}_{h=1}^{H},
        \left\{\xi_{n,a,k}\right\}_{a\in\mathcal M_n},
        \bm s_{n,k}^{\mathrm{cur}},
        \zeta_{n,k},
        \eta_{n,k}^{CV},
        \bm c_n
    \right]^{\top},
    \label{eq:local_observation}
\end{equation}
where $\widetilde{\varpi}_{n,a,k}^{\,h}$ denotes a normalized movement-level pressure feature derived from the local pressure evaluation of constituent $h$; this actor-input normalization is distinct from the phase-pressure normalization $\widetilde P_{n,\phi}^{h}$ used by the lower-level Com-MP controller. The variable $\xi_{n,a,k}$ is the queued vehicle count associated with movement $a$. The current signal operation is represented by $\bm s_{n,k}^{\mathrm{cur}}$ and its elapsed green time $\zeta_{n,k}$. The local CV ratio $\eta_{n,k}^{CV}$ describes sensing support, while $\bm c_n$ encodes the compatible phase configuration of intersection $n$. Together, these features describe the local control context available to the actor. A common movement--phase representation with padding and validity masks is used to accommodate heterogeneous intersection structures \citep{zhu2026schedule}.

\subsubsection{Composition action}
\label{sec:action}

For the Q--CV implementation, the actor selects the scalar coefficient $\alpha_{n,k}$ in~\eqref{eq:qcv_composite_pressure}. Nearby coefficient values may induce the same lower-level signal phase because a phase changes only when the ordering of the weighted constituent preferences changes. We therefore use the discrete composition set
\[
    \mathcal A=\{0,0.25,0.50,0.75,1\}.
\]
The endpoints recover the two constituent preferences: $\alpha_{n,k}=0$ gives Q-MP and $\alpha_{n,k}=1$ gives CV-MP, while intermediate values represent mixed pressure preferences. Importantly, the actor does not select a signal phase directly. Its action modifies only the relative contribution of the two normalized constituent pressures, and the phase is subsequently determined by the lower-level Com-MP maximization in~\eqref{eq:qcv_com_mp_control}. This restricted action interface is deliberate: the learning layer adapts only the relative emphasis of established MP preferences, so a learned action cannot directly prescribe a signal phase outside the structured Com-MP rule.

Sensing availability is handled outside the learned action. When only one constituent state is usable, the effective composition is deterministically assigned to the corresponding available constituent, so the controller reverts to Q-MP or CV-MP rather than extrapolating a learned mixture from unavailable information. When neither constituent state is usable, a predefined operational fallback is applied. These rules preserve a fixed composition-action space and provide a structured response to unavailable sensing information.

\subsubsection{Local reward}
\label{sec:reward}

The reward evaluates the traffic effect of the selected composition rather than the magnitude of the pressure objective itself. Let $Q_{n,k}$ denote the aggregate queued vehicle count associated with intersection $n$ at decision epoch $t_k$. The local reward is
\begin{equation}
    \mathcal U_{n,k}
    =
    \omega_1
    \left(Q_{n,k-1}-Q_{n,k}\right)
    -
    \omega_2 Q_{n,k},
    \label{eq:local_reward}
\end{equation}
where $\omega_1,\omega_2>0$ are the weights on queue dissipation and residual congestion, respectively. The first term rewards reductions in local queues, whereas the second penalizes persistent accumulation. Because the reward is applied after the lower-level MP decision is executed, the actor learns whether a composition improves traffic performance rather than merely increasing the instantaneous composite pressure.

\subsection{Phase-equivalent policy learning}
\label{sec:rl_training}

We use MAPPO to optimize the shared composition policy. Training samples collected from all intersections update the common actor $\pi_\psi$ and the centralized critic, while generalized advantage estimation (GAE) is used to construct the advantage estimates. The key modification arises from the structure of the Com-MP action: different composition coefficients can induce the same lower-level phase under the same traffic state. Treating such coefficients as behaviorally distinct in the PPO update would assign separate likelihoods to actions that produce the same immediate signal decision.

Let $\phi_{n,k}(\alpha)$ denote the phase ultimately selected by the lower-level Com-MP rule when coefficient $\alpha\in\mathcal A$ is applied under the current state and operational constraints. For a feasible phase $\phi\in\mathcal P_n$, define the state-dependent phase-equivalence group as $\mathcal A_{n,k}(\phi)\triangleq\{\alpha\in\mathcal A:\phi_{n,k}(\alpha)=\phi\}$. Coefficients in the same group are therefore equivalent with respect to the lower-level signal action at that decision epoch.

The actor probability of inducing phase $\phi$ is obtained by marginalizing the coefficient probabilities over the corresponding equivalence group:
\begin{equation}
    \pi_{\psi}^{\mathrm{ph}}
    \left(
        \phi
        \mid
        \bm o_{n,k}
    \right)
    \triangleq
    \sum_{\alpha\in\mathcal A_{n,k}(\phi)}
    \pi_{\psi}
    \left(
        \alpha
        \mid
        \bm o_{n,k}
    \right).
    \label{eq:grouped_phase_probability}
\end{equation}
For the sampled coefficient $\alpha_{n,k}$, let $\phi_{n,k}^{*}=\phi_{n,k}(\alpha_{n,k})$ be the resulting lower-level phase. The PPO likelihood ratio is then formed at the phase-equivalence level:
\begin{equation}
    \chi_{n,k}(\psi)
    =
    \frac{
        \pi_{\psi}^{\mathrm{ph}}
        \left(
            \phi_{n,k}^{*}
            \mid
            \bm o_{n,k}
        \right)
    }{
        \pi_{\psi_{\mathrm{old}}}^{\mathrm{ph}}
        \left(
            \phi_{n,k}^{*}
            \mid
            \bm o_{n,k}
        \right)
    },
    \label{eq:grouped_ppo_ratio}
\end{equation}
where $\psi_{\mathrm{old}}$ denotes the actor parameters used to collect the training sample. Using this grouped ratio, the clipped actor objective is
\begin{equation}
    L^{\mathrm{actor}}(\psi)
    =
    \mathbb E
    \left[
        \min
        \left(
            \chi_{n,k}(\psi)\widehat A_{n,k},
            \operatorname{clip}
            \left(
                \chi_{n,k}(\psi),
                1-\epsilon_{\mathrm{clip}},
                1+\epsilon_{\mathrm{clip}}
            \right)
            \widehat A_{n,k}
        \right)
    \right],
    \label{eq:grouped_ppo_objective}
\end{equation}
where $\widehat A_{n,k}$ is the estimated advantage and $\epsilon_{\mathrm{clip}}>0$ is the PPO clipping parameter. The grouping therefore changes only how the actor likelihood is evaluated; the underlying MAPPO optimization remains otherwise standard. It ensures that coefficients producing the same signal action receive a common policy-level learning signal while retaining their predefined interpretation as different composition weights \citep{pmlr-v101-asadi19a}.

If all candidate coefficients induce the same phase at a decision epoch, the phase probability in~\eqref{eq:grouped_phase_probability} is one independently of the coefficient distribution, so the transition contains no information for distinguishing among composition actions. Such transitions may still contribute to value-function learning but do not provide a discriminative actor signal. During decentralized execution, each intersection simply samples $\alpha_{n,k}$ from~\eqref{eq:local_actor} and passes it to the lower-level Com-MP controller; neither the centralized critic nor network-level observations are required online.

\section{Evaluation}
\label{sec:evaluation}

We evaluate Com-MP in microscopic traffic simulation using Simulation of Urban MObility (SUMO) \citep{lopez2018microscopic}. The experiments address four questions corresponding to the proposed learning architecture: whether context-dependent composition improves traffic performance, whether the learned composition exhibits interpretable dependence on local traffic and sensing conditions, whether the controller remains effective as sensing availability and quality change, and whether the learned local composition rule transfers to an unseen road network.

\subsection{Experimental setup}
\label{sec:experiment_setup}

The main experiment uses a real-world urban road network in Qinzhou, China, containing 45 signalized intersections, as shown in Fig.~\ref{FIG:road_network}(a). The longest link is over 1000~m, and the shortest is less than 200~m, producing substantial geometric and storage heterogeneity across intersections. Right-turn movements are not signal-controlled in the experiment.

\begin{figure}
    \centering
    \includegraphics[width=\textwidth]{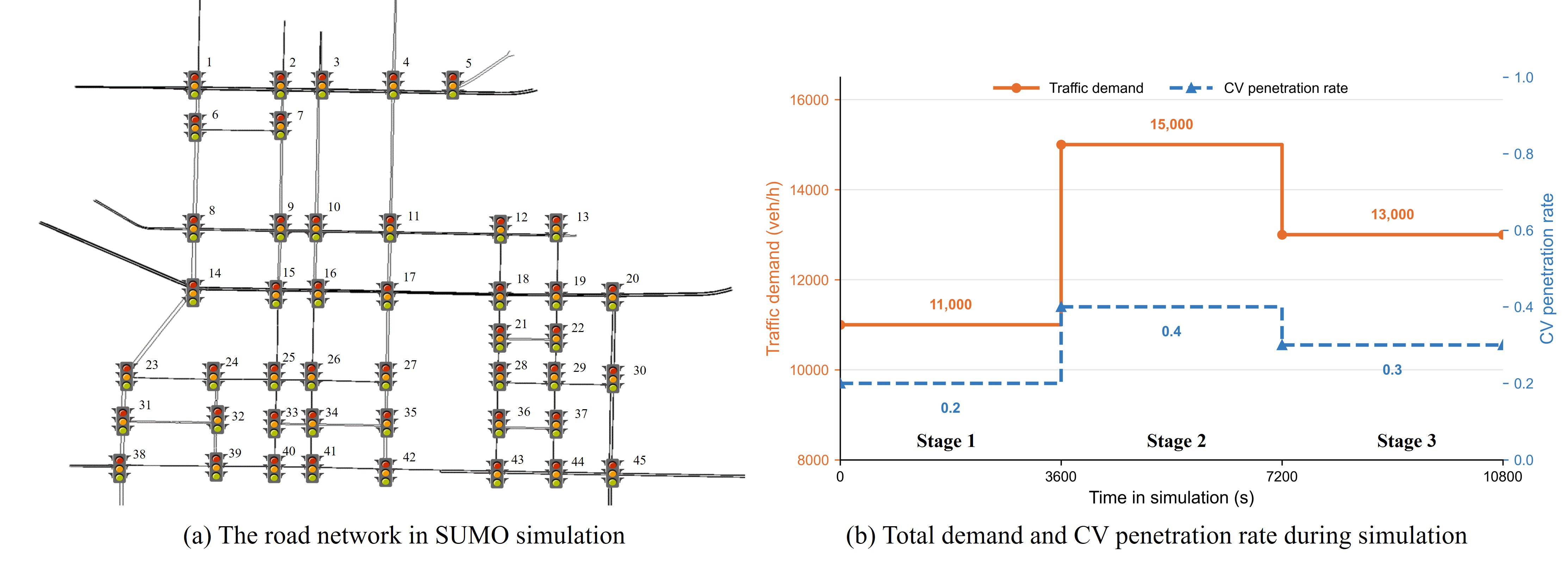}
    \caption{Experimental road network and time-varying traffic settings in Qinzhou, China.}
    \label{FIG:road_network}
\end{figure}

To create dynamically changing traffic and sensing conditions, the simulation lasts 10,800~s and is divided into three stages, as shown in Fig.~\ref{FIG:road_network}(b). The total traffic demand is 11,000~veh/h in Stage~1, 15,000~veh/h in Stage~2, and 13,000~veh/h in Stage~3, while the CV penetration rate changes from 0.2 to 0.4 and then to 0.3. Five random seeds are tested for each control method. The main experiment assumes 100\% detector coverage with no missed detections. Additional experiments investigate the effects of detector coverage, missed detection rates, and CV penetration in Section~\ref{sec:sensing_availability}.

We compare Com-MP with its two constituent controllers:
\begin{itemize}
    \item \textit{Q-MP} \citep{varaiya2013max}: the detector-supported queue-based MP controller defined in Section~\ref{sec:representative_controllers}, which constructs pressure from queue-related traffic states and emphasizes the dissipation of heavily loaded movements;
    \item \textit{CV-MP} \citep{tan2026cvmp}: the CV-supported travel-time-based MP controller defined in Section~\ref{sec:representative_controllers}, which constructs pressure from CV-derived travel-time states and emphasizes accumulated temporal cost; and
    \item \textit{Com-MP}: the proposed controller, which combines the Q-MP and CV-MP pressures using the intersection-wise coefficient $\alpha_{n,k}$ selected by the learned policy in Section~\ref{sec:rl_composition}.
\end{itemize}

All three MP controllers use a 15-s signal decision interval. The maximum red time is set as $T_{n,\phi}^{red,\max}=180~s$. A 3-s yellow interval is inserted whenever the selected phase changes.

For MAPPO training, the batch size is 512. The learning rate for both actor and critic follows a cosine schedule from $3\times10^{-4}$ to $3\times10^{-5}$. The value-loss coefficient is 0.5, while the entropy-loss coefficient decreases from 0.015 to 0.001 during training. The GAE factor is 0.97, the PPO clip ratio is 0.2, the maximum gradient norm is 0.5, and each PPO update uses five epochs. The reward weights in~\eqref{eq:local_reward} are $\omega_1=1.00$ and $\omega_2=0.10$. Both actor and critic use fully connected hidden layers with dimension 256 and are optimized with Adam \citep{kingma2014adam}. Each training episode lasts 3 hours. The detailed implementation is available at \url{https://github.com/bryan-zjc/Com-MP}.

The following measures of effectiveness are reported. \textit{Average delay} is the mean delay experienced by all generated vehicles, including delay accumulated inside the road network and waiting delay of vehicles temporarily unable to enter because of congestion. \textit{Average speed} is the mean vehicle travel speed within the network. \textit{Average stop count} is the mean number of stops experienced per vehicle. \textit{Throughput} is the number of vehicles completing their trips per unit time. 

\subsection{Experiment under time-varying demand and CV penetration rate}
\label{sec:dynamic_experiment}

We first examine whether the learned context-dependent composition improves network-level traffic performance under the time-varying demand and CV-penetration setting in Fig.~\ref{FIG:road_network}. The training reward trajectory is reported separately in Appendix~\ref{app:training_convergence}.

\begin{table}
    \centering
    \caption{Network-level performance of the constituent, fixed-composition, and adaptive control strategies in the Qinzhou network.}
    \label{tab:traffic_efficiency}
    {%
    \begin{tabular}{lcccc}
        \toprule
        Method & Average delay (s) & Average speed (km/h) & Average stop count & Throughput (veh/h) \\
        \midrule
        Q-MP   & 193.03 $\pm$ 5.30 & 28.37 $\pm$ 0.21 & 4.56 $\pm$ 0.14 & 12535.67 $\pm$ 18.19\\
        CV-MP  & 178.27 $\pm$ 4.60 & 28.78 $\pm$ 0.21 & \underline{4.36} $\pm$ 0.10 & 12557.53 $\pm$ 8.90\\
        Fixed weight ($\alpha=0.25$)  & 186.63 $\pm$ 5.05 & 28.65 $\pm$ 0.20 & 4.48 $\pm$ 0.12 & 12543.60 $\pm$ 16.56\\
        Fixed weight ($\alpha=0.50$)  & 178.58 $\pm$ 6.68 & 28.95 $\pm$ 0.26 & 4.38 $\pm$ 0.15 & 12552.33 $\pm$ 10.65\\
        Fixed weight ($\alpha=0.75$)  & \underline{175.15} $\pm$ 4.29 & \underline{29.06} $\pm$ 0.18 & \textbf{4.33} $\pm$ 0.09 & \underline{12561.47} $\pm$ 9.92\\
        Com-MP & \textbf{162.92} $\pm$ 3.77 & \textbf{29.77} $\pm$ 0.17 & 4.41 $\pm$ 0.09 & \textbf{12575.20} $\pm$ 6.93\\
        \bottomrule
    \end{tabular}
    }
\end{table}

Table~\ref{tab:traffic_efficiency} summarizes the network-level performance in the Qinzhou road network. Bold entries denote the best result, while underlined entries denote the second best across all compared methods. Focusing first on the constituent controllers, Com-MP reduces average delay by 15.6\% and 8.6\% relative to Q-MP and CV-MP, respectively. It also achieves the highest average speed and throughput. Although CV-MP yields a slightly lower average stop count, Com-MP provides the best overall network performance among Q-MP, CV-MP, and Com-MP. These results show that context-dependent composition can outperform either constituent pressure preference alone.

\begin{figure}
    \centering
    \includegraphics[width=0.8\textwidth]{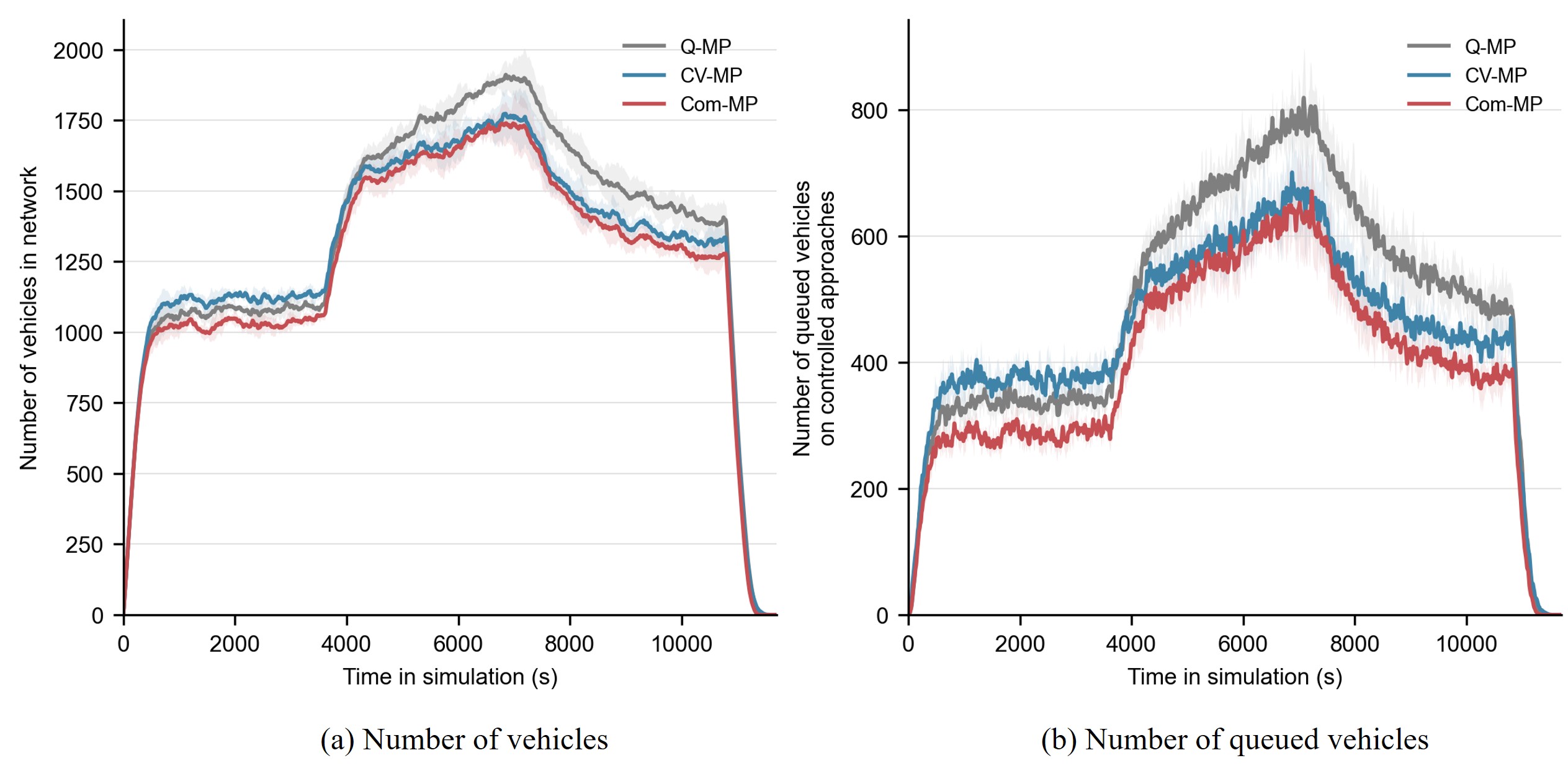}
    \caption{Time-series network performance under time-varying demand and CV penetration.}
    \label{FIG:network_performance}
\end{figure}

Fig.~\ref{FIG:network_performance} further illustrates the temporal evolution of network traffic conditions under time-varying demand and CV penetration. During Stage~1, when traffic demand and CV penetration are relatively low, CV-MP maintains slightly more vehicles and queued vehicles than Q-MP, whereas Com-MP achieves the lowest network accumulation and queue levels. During Stage~2, as both traffic demand and CV penetration increase, the numbers of vehicles and queued vehicles rise considerably under all three methods. Q-MP exhibits the greatest accumulation, while CV-MP maintains lower queue levels and Com-MP consistently achieves the best performance. During Stage~3, as traffic demand and CV penetration decrease, network accumulation and queue levels gradually decline under all three methods. Nevertheless, Q-MP continues to exhibit the highest accumulation and queue levels, followed by CV-MP, whereas Com-MP maintains the lowest levels throughout most of this stage. Overall, Com-MP maintains lower network accumulation and queue levels as the traffic and sensing contexts evolve, with particularly noticeable advantages during the higher-demand stage.

\subsection{Interpretable composition behavior}
\label{sec:heterogeneity_composition}

\begin{figure}
    \centering
    \includegraphics[width=0.8\textwidth]{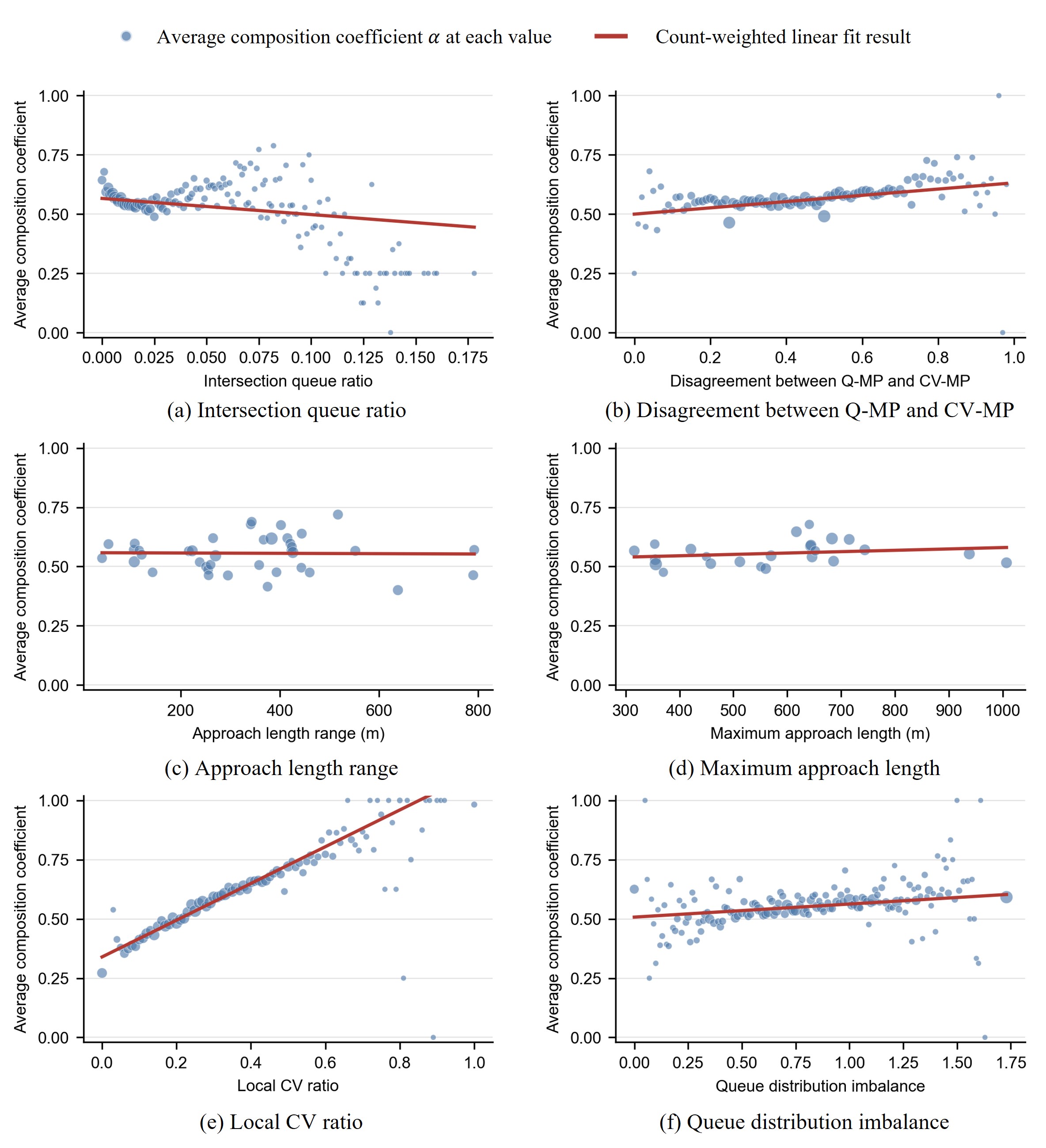}
    \caption{Mean selected composition coefficient under different factors.}
    \label{FIG:alpha_trend}
\end{figure}

The composition coefficient has an explicit control interpretation: in the Q--CV implementation, it directly represents the relative influence of CV-MP versus Q-MP. This makes it possible to examine the learned adaptation in terms of physically meaningful traffic and sensing factors rather than treating the policy action as an opaque phase command. Fig.~\ref{FIG:alpha_trend} relates the selected coefficient $\alpha$, used here as shorthand for the intersection-epoch coefficient $\alpha_{n,k}$, to local traffic conditions, sensing support, and intersection geometry. For each factor value, the mean selected $\alpha$ is reported, with scatter-point size proportional to the number of corresponding decisions and the fitted trend weighted by these counts. The results show that the learned composition responds most clearly to dynamic traffic and sensing conditions, while static geometric factors have a weaker association.

\paragraph{Traffic-state dependence.}
The intersection queue ratio measures queue accumulation relative to available storage. As shown in Fig.~\ref{FIG:alpha_trend}(a), the selected $\alpha$ decreases as the queue ratio increases, indicating a shift toward Q-MP under more heavily queued conditions. Fig.~\ref{FIG:alpha_trend}(f) further shows a mild increase in $\alpha$ as the queue distribution becomes more imbalanced across approaches. Together, these patterns indicate that the learned policy changes the constituent emphasis with the local congestion state rather than applying a fixed pressure preference.

\paragraph{Constituent disagreement.}
The disagreement in Fig.~\ref{FIG:alpha_trend}(b) is the mean absolute difference between the normalized Q-MP and CV-MP phase pressures over the candidate phases. A larger value therefore indicates less consistent phase preferences between the two constituents. The positive relationship between disagreement and $\alpha$ shows that constituent disagreement is associated with systematic changes in the learned composition, with greater disagreement corresponding to increased influence of the travel-time-based preference.

\paragraph{Sensing availability and geometry.}
The strongest sensing-related pattern appears in Fig.~\ref{FIG:alpha_trend}(e): the selected $\alpha$ increases almost monotonically with the local CV ratio, so richer CV observations lead to greater reliance on CV-MP, whereas sparse CV observations shift the composition toward Q-MP. In contrast, the approach-length range and maximum approach length in Fig.~\ref{FIG:alpha_trend}(c)--(d) show only weak relationships with $\alpha$. The learned composition therefore shows a stronger relationship with dynamic traffic and sensing context than with static intersection geometry.

\paragraph{Adaptive versus fixed composition.}
To separate the benefit of pressure composition itself from the additional benefit of context-dependent adaptation, we compare Com-MP with network-wide fixed coefficients drawn from the same action set used by the learned policy. The fixed strategies keep $\alpha$ unchanged across all intersections and decision epochs, whereas Com-MP selects $\alpha_{n,k}$ from the prevailing local traffic and sensing context.
The fixed-composition results in Table~\ref{tab:traffic_efficiency} further show that the effectiveness of combining the two constituent pressure formulations depends on the selected coefficient. Among the fixed strategies, $\alpha=0.75$ performs best overall. Adaptive Com-MP further reduces average delay from 175.15~s to 162.92~s, or approximately 7.0\%, while also increasing average speed and throughput. This comparison indicates that context-dependent adaptation provides an additional benefit beyond using a fixed composition.

\subsection{Intersection heterogeneity}
\label{sec:intersection_heterogeneity}
This section examines how the relative usefulness of the constituent MP controllers changes across local traffic contexts. We compare the intersection-level performance of Q-MP and CV-MP across traffic stages and evaluate whether Com-MP can provide further improvements beyond the locally better constituent.

\begin{figure}
    \centering
    \includegraphics[width=\textwidth]{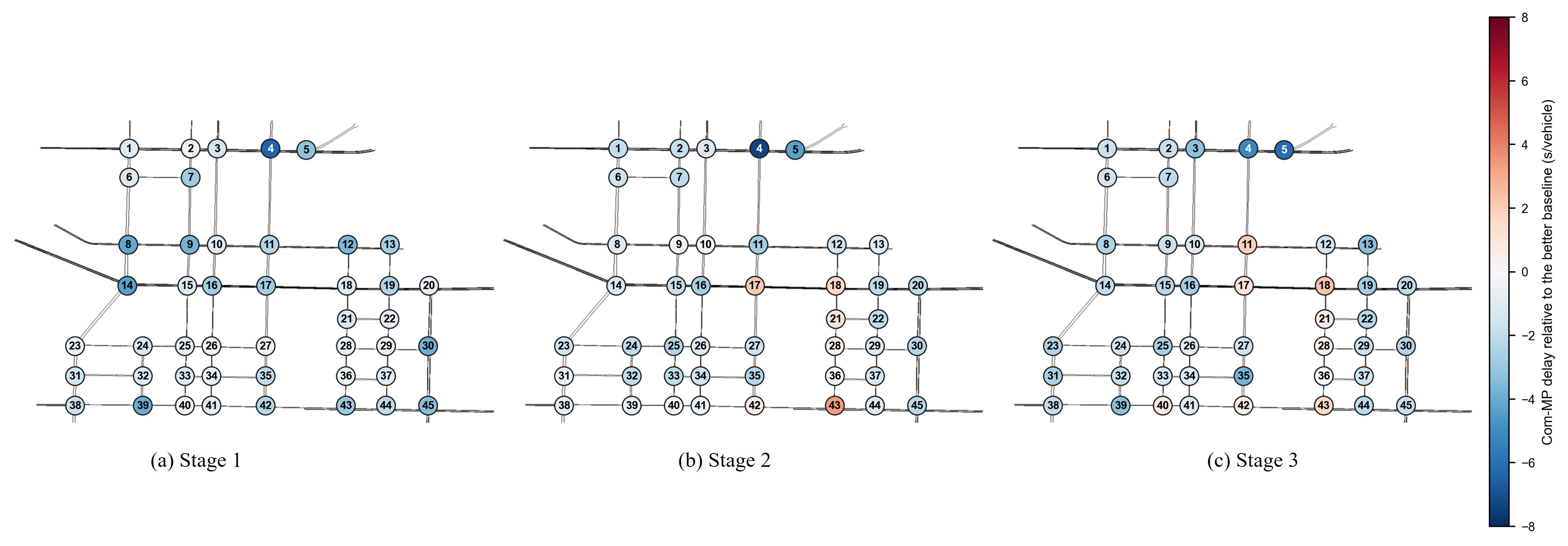}
    \caption{Spatial distribution of Com-MP delay differences relative to the better constituent controller across demand stages.}
    \label{FIG:perfomance_difference}
\end{figure}

Fig.~\ref{FIG:perfomance_difference} first shows, for each intersection and traffic stage, the delay difference between Com-MP and the better-performing constituent controller, where the better constituent is the one with the lower local delay between Q-MP and CV-MP. Negative values therefore indicate that Com-MP achieves lower delay than both constituent controllers. Across the three stages, Com-MP produces lower delay than the better constituent at most intersections, while the spatial pattern of the improvement changes with demand and CV penetration. These results indicate that the additional benefits of Com-MP vary across intersections and traffic stages.
Representative intersection-level examples illustrating these spatial and temporal changes are provided in Appendix~\ref{app:representative_intersections}.

\begin{figure}
    \centering
    \includegraphics[width=0.8\textwidth]{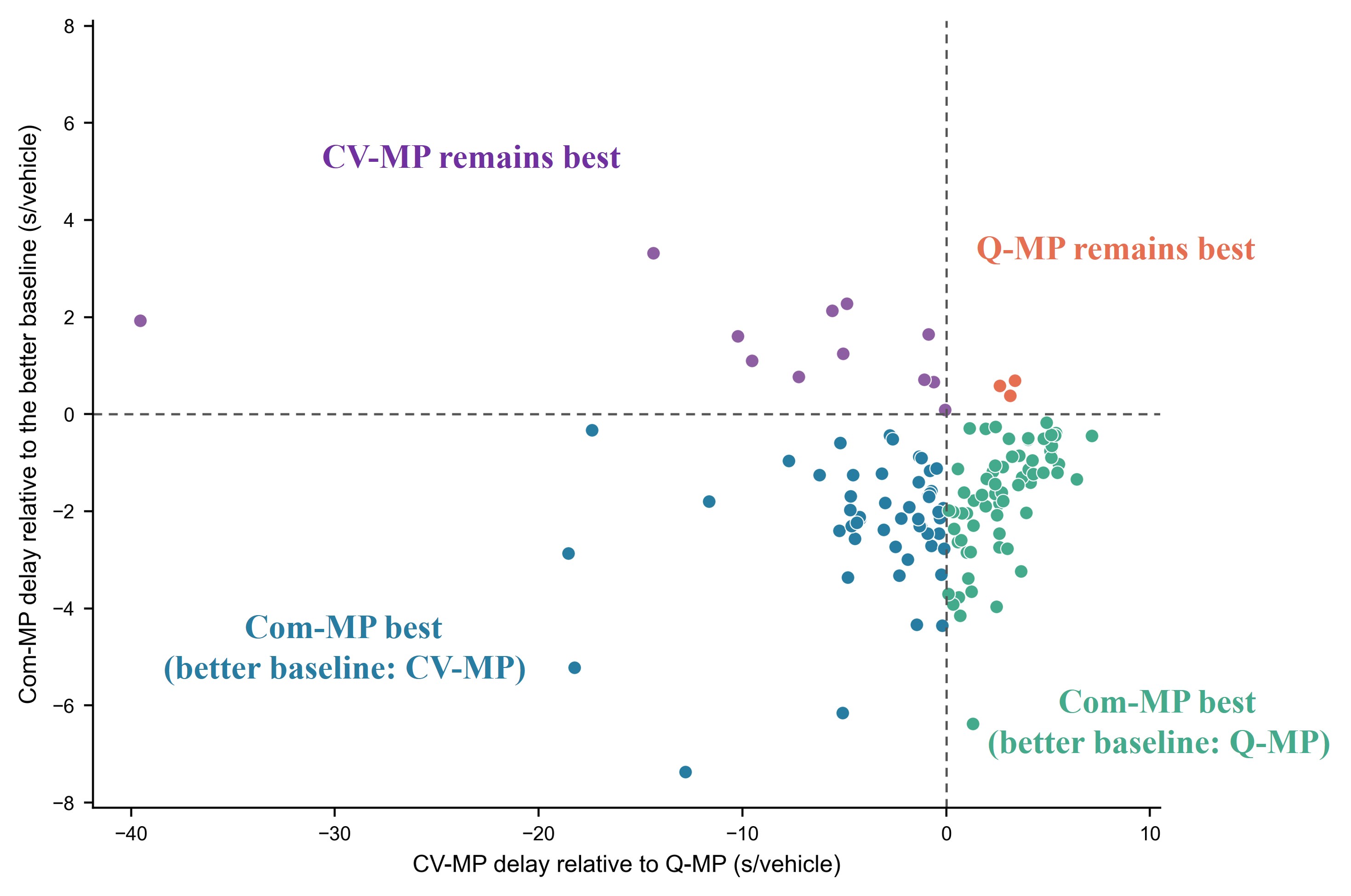}
    \caption{Local intersection delay comparison among Q-MP, CV-MP, and Com-MP.}
    \label{FIG:prefer_scatter}
\end{figure}

Fig.~\ref{FIG:prefer_scatter} further examines the relative performance of the three control methods across all 45 intersections and three traffic stages. The horizontal axis represents the intersection-level delay difference between CV-MP and Q-MP, where negative values indicate that CV-MP achieves lower delay, while positive values indicate that Q-MP performs better. The vertical axis represents the delay difference between Com-MP and the better-performing constituent controller, with negative values indicating that Com-MP outperforms both Q-MP and CV-MP. The observations are distributed on both sides of the vertical zero line, showing that the locally better constituent changes across intersections and traffic stages. More importantly, most observations fall below the horizontal zero line, showing that Com-MP further reduces delay relative to the locally better constituent in most cases. This pattern occurs whether Q-MP or CV-MP is locally preferable. These results provide direct spatial and temporal evidence that constituent usefulness changes with traffic context and that the learned composition can provide additional local benefits beyond either constituent alone.

\subsection{Robustness under heterogeneous and degraded sensing}
\label{sec:sensing_availability}
The preceding experiments show that the relative usefulness of Q-MP and CV-MP varies with local traffic conditions and that Com-MP adjusts their contributions accordingly. We next examine the empirical robustness of this adaptation as sensing availability and quality change. The two constituents rely on different information sources: Q-MP uses detector-supported queue information, whereas CV-MP uses vehicle-level trajectory information from CVs. We therefore vary detector coverage, detector missed detections, and CV penetration to evaluate whether Com-MP remains effective under heterogeneous and degraded sensing conditions.

\subsubsection{Different detector coverage rates}
\label{sec:detector_coverage}

Fixed detectors are not uniformly available in practical road networks. We therefore vary the proportion of signalized intersections equipped with detectors from 10\% to 90\%, while keeping the remaining experimental settings unchanged. At detector-equipped intersections, queue-related information is available for Q-MP. At intersections without detector information, Com-MP relies on the available CV-MP constituent according to the sensing rule in Section~\ref{sec:action}. CV-MP is included as a reference because it does not use fixed-detector information.

\begin{figure}
    \centering
    \includegraphics[width=0.8\textwidth]{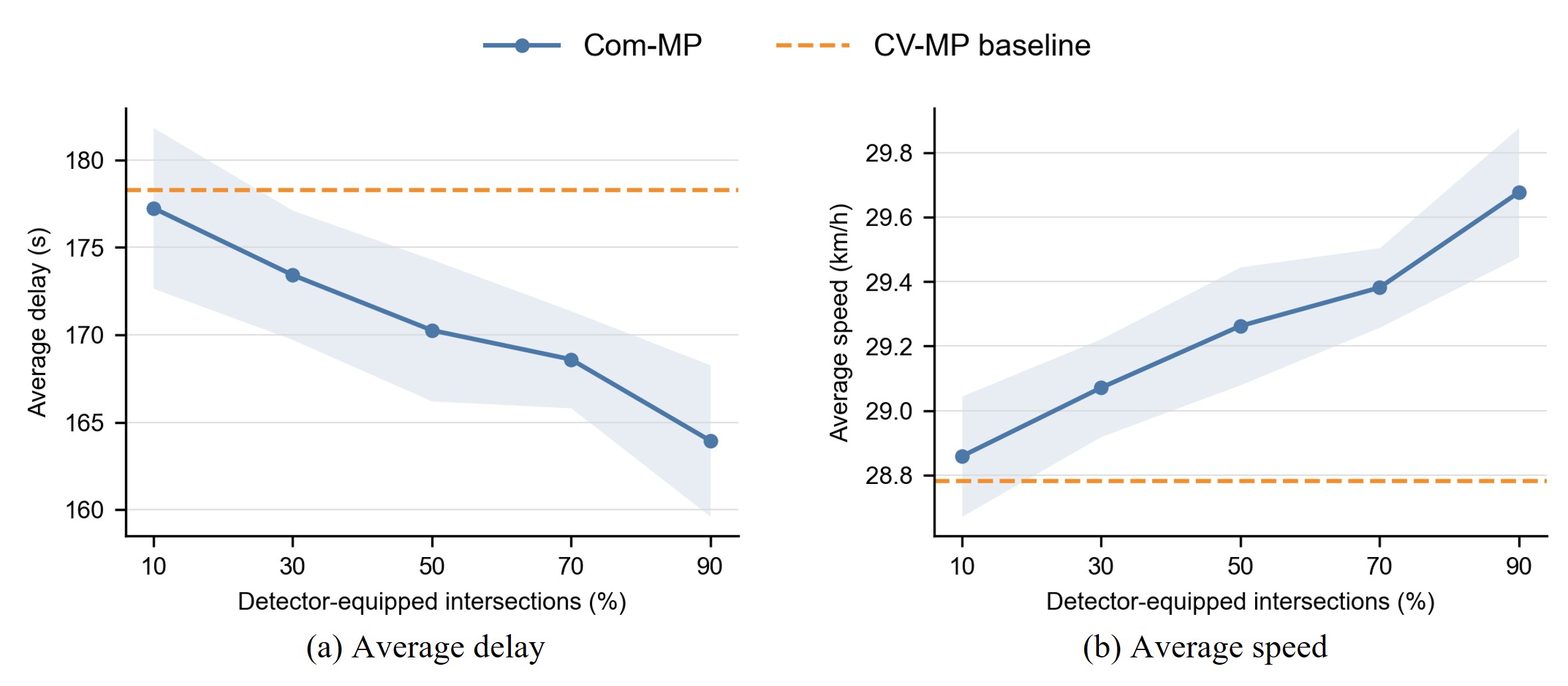}
    \caption{Network performance under different fixed-detector coverage rates.}
    \label{FIG:detector_coverage}
\end{figure}

As shown in Fig.~\ref{FIG:detector_coverage}, the performance of Com-MP improves consistently as detector coverage increases. The dashed CV-MP baseline represents the average performance of CV-MP across the five random seeds. Specifically, the average delay decreases from approximately 177.2~s at 10\% coverage to 163.9~s at 90\%, while the average speed increases from about 28.86 to 29.68~km/h. CV-MP remains unchanged because it relies only on CV-based traffic information. Notably, even when only 10\% of intersections are equipped with detectors, Com-MP achieves slightly lower average delay and higher average speed than CV-MP, indicating that even limited detector-supported queue information can complement the CV-based constituent. As detector coverage increases, more intersections can utilize both queue- and CV-based information, leading to further improvements in network performance. These results show that Com-MP remains effective under partial detector coverage while benefiting from greater availability of queue information.

\subsubsection{Different missed detection rates}
\label{sec:missed_detection}

In practice, traffic detectors may fail to detect all vehicles because of sensing and communication imperfections. To examine the robustness of Com-MP to such imperfect queue observations, we design missed detection conditions using two parameters: the missed detection lane ratio, which denotes the proportion of lanes affected by missed detections, and the missed detection rate, which denotes the proportion of vehicles that are not detected on an affected lane. Five random seeds are tested for each control method. When varying the vehicle missed detection rate, the missed detection lane ratio is fixed at 30\%; conversely, when varying the missed detection lane ratio, the vehicle missed detection rate is fixed at 30\%.

\begin{figure}
    \centering
    \includegraphics[width=0.8\textwidth]{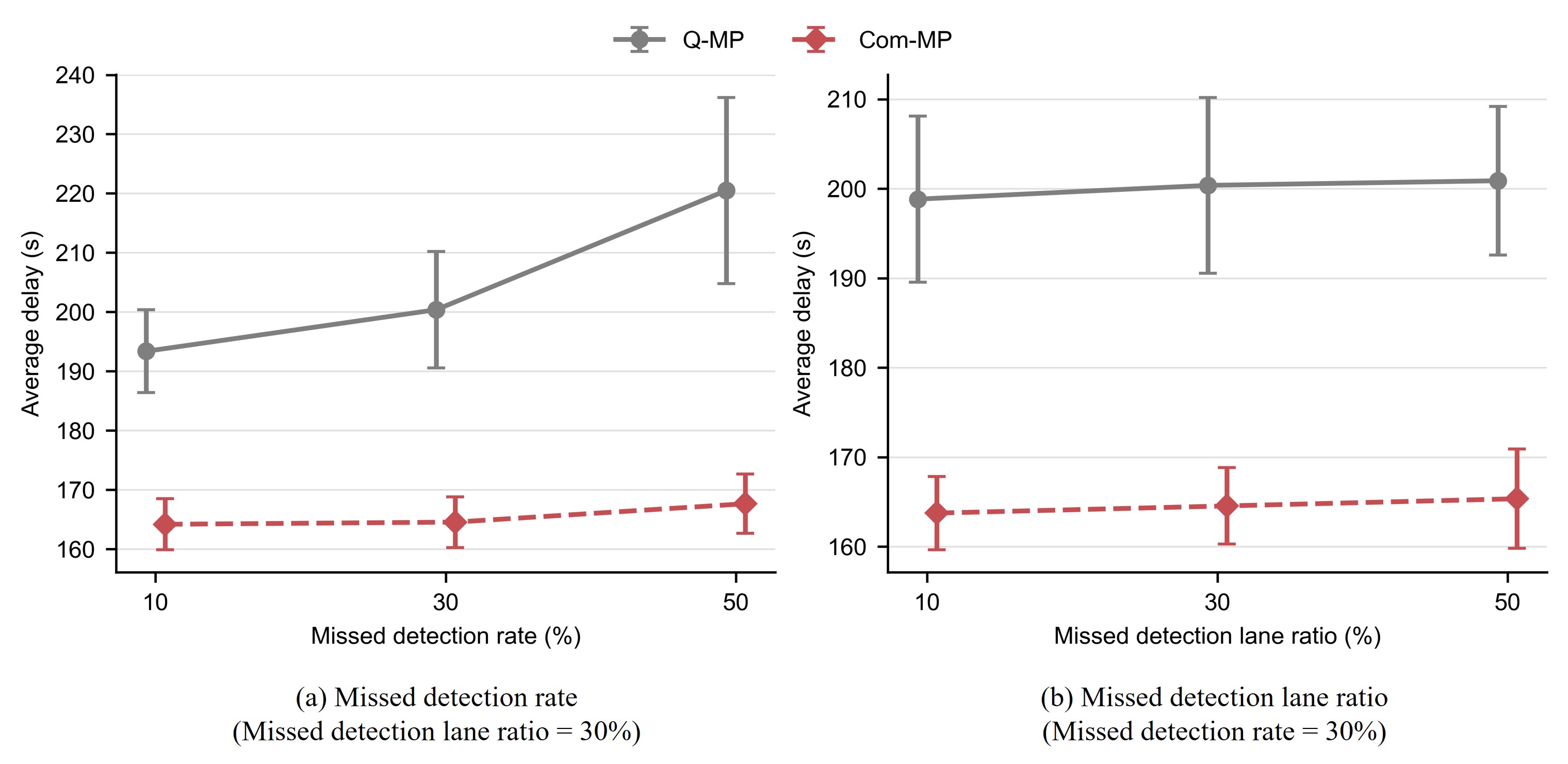}
    \caption{Average delay under different detector miss-detection conditions.}
    \label{FIG:missed_detection}
\end{figure}

Fig.~\ref{FIG:missed_detection}(a) shows that increasing the vehicle missed detection rate has a pronounced adverse effect on Q-MP. Its average delay increases from approximately 193.5~s at a 10\% missed detection rate to 220.5~s at 50\%. In contrast, Com-MP is much less sensitive, with average delay increasing only from approximately 164.2~s to 167.7~s over the same range. This shows that Com-MP is substantially less sensitive to degraded detector-supported queue measurements while retaining the complementary CV-based traffic information.

Fig.~\ref{FIG:missed_detection}(b) shows the sensitivity to the proportion of lanes affected by missed detection. When the missed detection lane ratio increases from 10\% to 50\%, the average delay of both methods increases only slightly, with Com-MP consistently outperforming Q-MP. Compared with the missed detection rate, the proportion of affected lanes has a much smaller influence on control performance within the tested ranges. This suggests that the severity of missed detection on affected lanes is more detrimental than the spatial extent of detector degradation under the considered experimental settings.

Overall, Com-MP shows greater robustness to detector miss-detection than Q-MP under both degradation dimensions. The results further show that the severity of missed detection on affected lanes has a stronger impact than the spatial proportion of affected lanes within the tested ranges. Together with the sensing-availability rule in Section~\ref{sec:action}, the results illustrate two complementary responses to sensing uncertainty: when one constituent state is unavailable, the controller reverts to the available MP constituent; when multiple constituents remain usable but one sensing source is degraded, adaptive composition can reduce dependence on that degraded preference. This structured combination reduces reliance on a single sensing source and maintains relatively stable performance under imperfect detector observations.

\subsubsection{Higher CV penetration rates}
\label{sec:higher_cv}

In this section, we increase the CV penetration rates in the three traffic stages from $(0.2,0.4,0.3)$ to $(0.4,0.7,0.5)$ while keeping all other settings unchanged to evaluate the sensitivity of different control methods to CV information availability.

\begin{figure}
    \centering
    \includegraphics[width=0.8\textwidth]{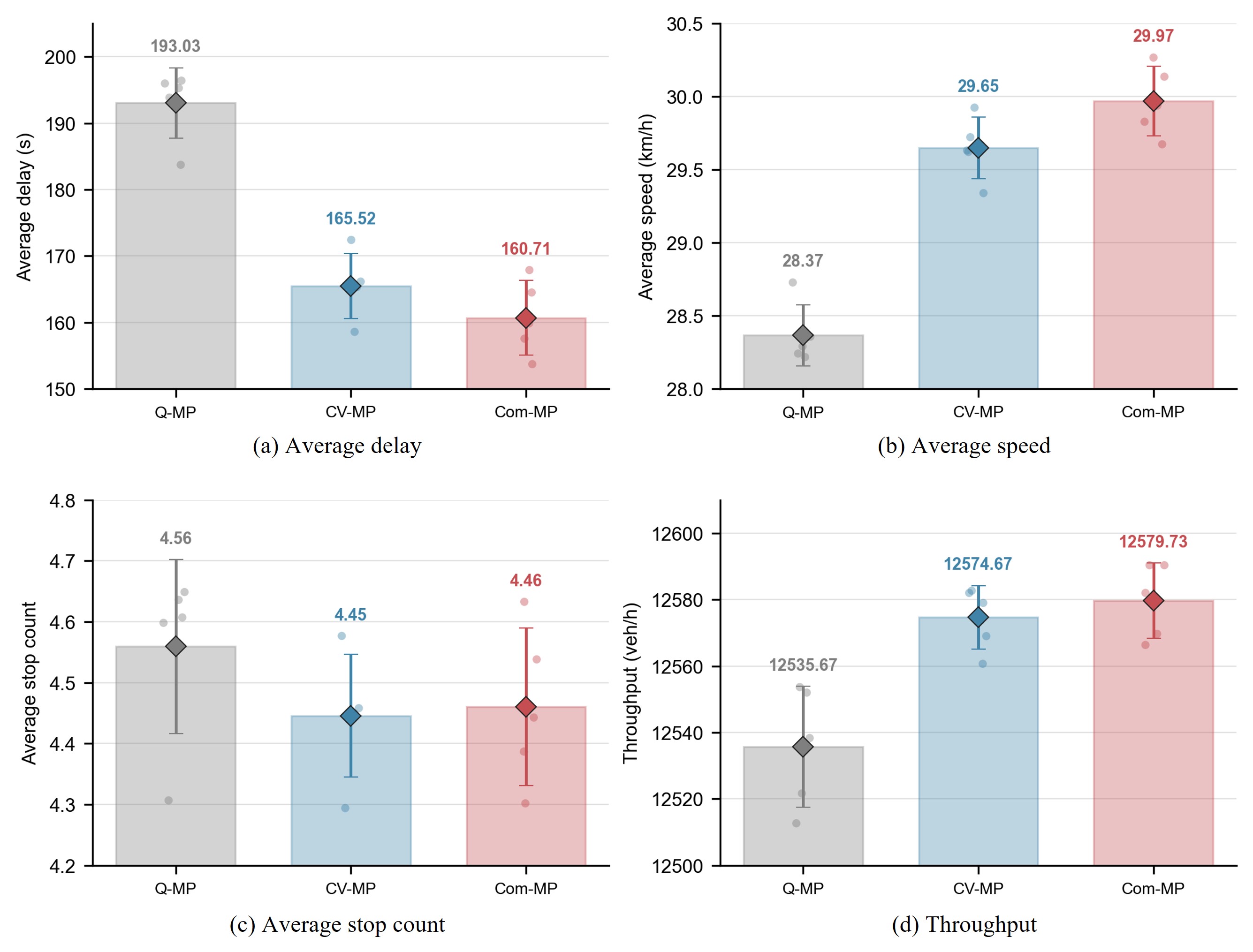}
    \caption {Traffic results under higher CV penetration rates.}
    \label{FIG:traffic_efficiency_higher_cv}
\end{figure}

As shown in Fig.~\ref{FIG:traffic_efficiency_higher_cv}, under the higher-CV setting, CV-MP improves markedly and outperforms Q-MP across all four performance measures. Its average delay decreases to 165.52~s, compared with 193.03~s for Q-MP, suggesting that richer CV observations improve the effectiveness of the travel-time-based pressure formulation. Com-MP further reduces the average delay to 160.71~s. It also achieves the highest average speed and throughput, at 29.97~km/h and 12,579.73~veh/h, respectively. Compared with Q-MP and CV-MP, Com-MP reduces average delay by approximately 16.7\% and 2.9\%, respectively. Although CV-MP achieves a slightly lower average stop count than Com-MP, the difference is marginal. These results indicate that Com-MP maintains its overall performance advantage under higher CV penetration.



Comparing the two CV-penetration settings further illustrates how the relative usefulness of the constituent controllers changes with sensing context. CV-MP benefits substantially from richer CV observations, and its average-delay gap relative to Com-MP decreases from 8.6\% in the original setting to 2.9\% under higher CV penetration. Com-MP nevertheless maintains the best overall traffic performance. Together with the composition patterns in Section~\ref{sec:heterogeneity_composition}, these results show that constituent usefulness changes with sensing availability and that context-dependent composition can accommodate this change.

\subsection{Cross-network transferability}
\label{sec:transferability}

To evaluate cross-network transfer, the actor trained on the Qinzhou network is deployed directly, without additional training, on an unseen real-world road network in Lianyungang City, Jiangsu Province, China, as shown in Fig.~\ref{FIG:lyg_network} (a). The network contains 28 signal-controlled intersections and includes shared lanes, heterogeneous phase configurations, and approaches consisting of multiple consecutive links. Vehicle arrivals follow Poisson processes. Over the 5,400~s simulation horizon, the expected network-wide demand increases from approximately 7,400~veh/h during the first 1,800~s to 9,700~veh/h between 1,800 and 3,600~s, before reaching a peak of 11,500~veh/h during the final 1,800~s, as shown in Fig.~\ref{FIG:lyg_network}(b). The CV penetration rate remains constant at 0.3 throughout the simulation. This progressively increasing demand creates increasingly challenging traffic conditions for evaluating the control methods. Each control method is evaluated using five matched random seeds. Q-MP and CV-MP are included as constituent-controller benchmarks. Additionally, a new Com-MP controller is trained directly on the Lianyungang network to provide a reference.

\begin{figure}
    \centering
    \includegraphics[width=\textwidth]{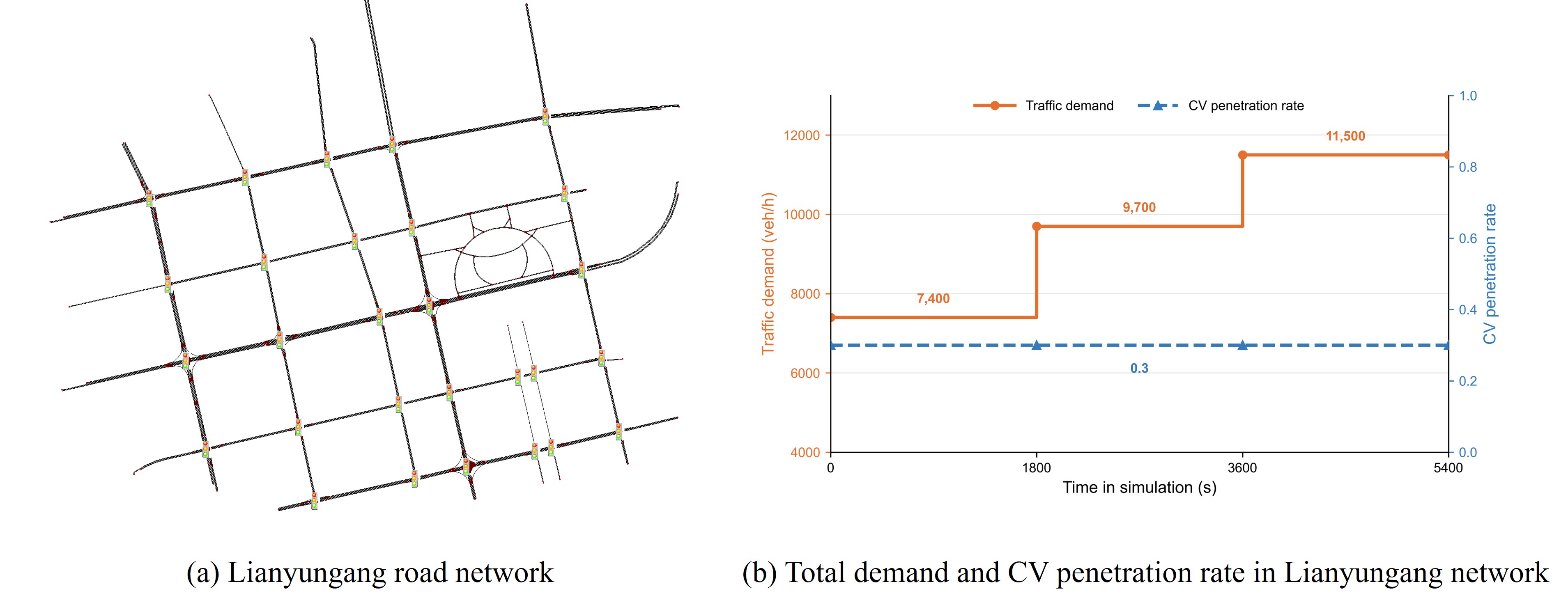}
    \caption{Unseen Lianyungang network used for the transfer experiment.}
    \label{FIG:lyg_network}
\end{figure}

Table~\ref{tab:traffic_efficiency_lyg} shows that the zero-shot transferred policy retains performance close to that of the policy trained directly on the Lianyungang network. The target-network-trained Com-MP achieves the best performance across all four measures, while the transferred Com-MP consistently ranks second. Compared with target-network training, zero-shot transfer increases average delay by only 2.4\%, reduces average speed by 0.5\%, and decreases throughput by approximately 0.2\%. Meanwhile, the transferred policy still outperforms both constituent controllers across all four measures, reducing average delay by 12.8\% and 2.8\% compared with Q-MP and CV-MP, respectively. These results indicate that the composition policy learned on the Qinzhou network can retain much of its control effectiveness when directly applied to a network with a different topology and traffic conditions.

\begin{table}
    \centering
    \caption{Results of different control methods in the unseen Lianyungang network.}
    \label{tab:traffic_efficiency_lyg}
    \resizebox{\textwidth}{!}{%
    \begin{tabular}{lcccc}
        \toprule
        Method & Average delay (s) & Average speed (km/h) & Average stop count & Throughput (veh/h)\\
        \midrule
        Q-MP & 191.90 $\pm$ 5.60 & 33.41 $\pm$ 0.15 & 4.99 $\pm$ 0.07 & 7616.67 $\pm$ 98.86\\
        CV-MP & 172.18 $\pm$ 2.63 & 34.18 $\pm$ 0.13 & 4.82 $\pm$ 0.04 & 7722.67 $\pm$ 85.50\\
        Com-MP (zero-shot transfer) & \underline{167.43} $\pm$ 2.77 & \underline{34.56} $\pm$ 0.04& \underline{4.80} $\pm$ 0.04 & \underline{7735.20} $\pm$ 82.44\\
        Com-MP (target-network trained) & \textbf{163.44} $\pm$ 3.40 & \textbf{34.74} $\pm$ 0.15 & \textbf{4.74} $\pm$ 0.09 & \textbf{7750.80} $\pm$ 90.11\\
        \bottomrule
    \end{tabular}
    }
\end{table}

The transferability is closely related to the architecture proposed in Sections~\ref{sec:rl_architecture} and~\ref{sec:rl_training}. The actor is shared across intersections, relies on local traffic and sensing observations, and outputs only the composition coefficient rather than topology-specific signal phases. The actual phase is still determined by the lower-level MP controller according to the local intersection structure. The learned policy therefore represents a context-dependent local rule for composing the two constituent MP preferences rather than a control policy tied to a specific network topology. The Lianyungang experiment provides empirical evidence of cross-network transferability, with the zero-shot policy retaining most of the performance of target-network retraining despite the change in network topology and traffic conditions.

\section{Conclusion and future work}
\label{sec:conclusion}

This study develops Com-MP as a theory-guided RL framework in which learning adapts the relative influence of heterogeneous MP control preferences while actual signal-phase selection remains governed by a structured decentralized MP mechanism. This restricted learning role preserves the analytical control structure rather than replacing it with directly learned phase actions. We establish that arbitrary admissible nonanticipative, intersection-specific, and time-varying compositions preserve strong queue stability for every demand in the strictly stabilizable region under the stated finite-storage traffic dynamics and operational signal constraints. The resulting stability certificate is independent of the particular learned composition policy, separating performance adaptation from the network-level stability property of the lower-level controller.

The experiments show that this separation does not prevent meaningful performance adaptation. The relative performance of Q-MP and CV-MP varies across intersections and traffic stages, while the learned composition changes systematically with local traffic and sensing context. At the network level, Com-MP reduces average delay by 15.6\% and 8.6\% relative to Q-MP and CV-MP, respectively, and by about 7.0\% relative to the best fixed composition. The explicit composition variable also provides interpretable evidence of the learned behavior: greater queue accumulation generally shifts the policy toward Q-MP, whereas richer CV availability increases the influence of CV-MP. These results show that the gain arises not only from combining complementary MP formulations but also from adapting their relative influence to prevailing conditions.

Com-MP further remains effective under heterogeneous and degraded sensing. Across variations in detector coverage, missed detections, and CV penetration, the controller maintains strong network performance by combining complementary sensing-dependent preferences and reverting to an available constituent when necessary. The learned local composition rule also exhibits cross-network transferability: when deployed without retraining on the unseen Lianyungang network, it achieves an average delay only 2.4\% higher than that of a policy trained directly on the target network while still outperforming both constituent controllers. Together, these results show that data-driven adaptation can be embedded in a structured controller while retaining a policy-independent stability guarantee. This provides a theory-guided pathway toward trustworthy learning-enabled traffic signal control.

Future work can extend the framework in several directions. First, the present implementation combines two representative constituents, Q-MP and CV-MP; incorporating a broader family of MP formulations with different state definitions and sensing requirements would further examine the generality of adaptive composition. Second, richer sensing and communication uncertainty, particularly time-varying data quality and communication delays, can be incorporated into constituent-state construction and composition learning. Third, the transferability of a shared composition policy should be evaluated across networks with different scales and topologies and under broader demand and sensing conditions to better characterize when context-dependent composition can be deployed without retraining.

\section*{Declaration of generative AI and AI-assisted technologies in the writing process}

During the preparation of this work, the authors used ChatGPT and Codex in order to improve the language and assist with programming. After using this tool/service, the authors reviewed and edited the content as needed and take full responsibility for the content of the publication.

\printcredits

\appendix

\section{Proofs for the Stability of Com-MP}
\label{app:stability_proof}

This appendix gives the detailed proofs and derivations supporting the propositions, lemmas, and stability theorem in Section~\ref{sec:stability_analysis}.

\subsection{Proof of Lemma~\ref{lem:strong_stability_criterion}}
\label{app:proof_strong_stability_criterion}

\begin{proof}
Applying the extended drift identity over $[0,T]$, taking expectations, and using $V(T)\ge0$ give
\begin{equation}
    \frac{1}{T}
    \int_0^T
    \mathbb E
    \left[
        \sum_{a\in\mathcal M_{\mathcal F}}z_a(t)
    \right]dt
    \le
    \frac{K}{\epsilon}
    +
    \frac{\mathbb E[V(0)]}{\epsilon T}.
    \label{eq:source_stability_criterion}
\end{equation}
The right-hand side remains finite as $T\rightarrow\infty$. Since every real movement has finite storage and therefore satisfies $z_a(t)\le z_a^{\max}<\infty$, the long-run average of the total network vehicle count is also finite. Hence Definition~\ref{def:strong_stability} holds.
\end{proof}

\subsection{Proof of Proposition~\ref{prop:normalized_composition_equivalence}}
\label{app:proof_normalized_composition}

Fix $n$ and $t_k$. If $S_{n,k}>0$, substituting~\eqref{eq:normalized_phase_pressure} into the composite phase pressure gives
\begin{align}
    \sum_h\alpha_{n,h,k}\widetilde P_{n,\phi}^{h}
    ={}&
    \sum_{h:\Gamma_{n,h,k}>0}
    \frac{\alpha_{n,h,k}}{\Gamma_{n,h,k}}
    P_{n,\phi}^{h}
    -
    \sum_{h:\Gamma_{n,h,k}>0}
    \frac{\alpha_{n,h,k}}{\Gamma_{n,h,k}}
    P_{n,\min}^{h}.
    \label{eq:app_normalized_expansion}
\end{align}
The second term is independent of $\phi$. Dividing the first term by the positive scalar $S_{n,k}$ therefore leaves its argmax unchanged and produces the coefficients in~\eqref{eq:effective_composition}. Hence~\eqref{eq:normalized_raw_equivalence} holds for any nonempty candidate subset $\mathcal P_n'\subseteq\mathcal P_n$.

If $S_{n,k}=0$, every constituent with $\alpha_{n,h,k}>0$ has $\Gamma_{n,h,k}=0$ and therefore assigns the same raw phase pressure to every phase. Both the normalized objective and the $\bm\theta_{n,k}=\bm\alpha_{n,k}$ raw objective are constant over $\phi$, so~\eqref{eq:normalized_raw_equivalence} again holds. In either case, the coefficients are nonnegative and sum to one, giving $\bm\theta_{n,k}\in\Delta_H$.

\subsection{Proof of Lemma~\ref{lem:composite_admissibility}}
\label{app:proof_composite_admissibility}

For a real movement $a$, the analytical coefficients are nonnegative and sum to one. Hence
\begin{align}
    0
    \le
    \bar w_a^{k}(t)
    &=
    \sum_h
    \theta_{\nu(a),h,k}w_a^h(t)
    \le
    \sum_h
    \theta_{\nu(a),h,k}W_a^h
    \le
    \bar W_a,
    \\
    0
    \le
    \bar G_a^{k}(t)
    &=
    \sum_h
    \theta_{\nu(a),h,k}G_a^h(t)
    \le
    \sum_h
    \theta_{\nu(a),h,k}G_a^{h,\max}
    \le
    \bar G_a.
\end{align}
The analytical coefficients are fixed within $I_k$. Applying~\eqref{eq:constituent_potential_drift},
\begin{align}
    \mathcal L_{\bm s}\bar G_a^{k}
    &=
    \sum_h
    \theta_{\nu(a),h,k}
    \mathcal L_{\bm s}G_a^h
    \nonumber\\
    &\le
    \sum_h
    \theta_{\nu(a),h,k}
    \left(
        w_a^h\mathcal L_{\bm s}z_a
        +
        K_a^h
    \right)
    \nonumber\\
    &\le
    \bar w_a^{k}\mathcal L_{\bm s}z_a
    +
    \bar K_a.
    \label{eq:app_composite_potential_drift}
\end{align}
All bounds are uniform over the induced analytical composition.

\subsection{Derivation of Lemma~\ref{lem:service_weighted_drift}}
\label{app:proof_service_drift}

For a fixed analytical composition during $I_k$, the real-movement part of~\eqref{eq:composite_lyapunov} satisfies
\begin{align}
    \mathcal L_{\bm s}
    \sum_{a\in\mathcal M_{\mathcal N}}
    \bar G_a^{k}
    \le
    \sum_{a\in\mathcal M_{\mathcal N}}
    \bar w_a^{k}
    \mathcal L_{\bm s}z_a
    +
    \sum_{a\in\mathcal M_{\mathcal N}}
    \bar K_a.
    \label{eq:app_real_potential_drift}
\end{align}

For a fictitious source movement $a$, the corresponding source Lyapunov term is $z_a^2/2$. Its extended drift can be written as
\begin{equation}
    \mathcal L_{\bm s}\frac12z_a^2
    =
    z_a\,\mathcal L_{\bm s}z_a
    +
    \delta_a^{\mathcal F},
    \qquad
    |\delta_a^{\mathcal F}|
    \le
    K_a^{\mathcal F},
    \label{eq:app_source_potential_drift}
\end{equation}
where $\delta_a^{\mathcal F}$ denotes the event-induced source-queue drift remainder and $K_a^{\mathcal F}$ is a finite uniform bound on its magnitude. This remainder contains the finite unit-jump correction and vanishes under an absolutely continuous fluid representation.

Summing the state-weighted count drifts over all movements and using the conditional count-drift identity~\eqref{eq:candidate_count_drift} gives
\begin{align}
    \sum_{a\in\mathcal M}
    \bar w_a^{k}\mathcal L_{\bm s}z_a
    &=
    (\bar{\bm w}^{k})^\top\bar{\bm\lambda}
    -
    (\bar{\bm w}^{k})^\top
    [\bm I-\bm R^\top(t)]
    \bar{\bm d}(\bm s,t)
    \nonumber\\
    &=
    (\bar{\bm w}^{k})^\top\bar{\bm\lambda}
    -
    \sum_{a\in\mathcal M}
    \bar p_a^{k}\bar d_a(\bm s,t).
    \label{eq:app_pressure_service_identity}
\end{align}
The extended drift of the real-movement quadratic term $\frac{\kappa}{2}\sum_{a\in\mathcal M_{\mathcal N}}z_a^2$ is uniformly bounded above because both real-movement counts and physical movement flows are bounded. Collecting this term, the constituent Lyapunov-term bounds, source jump corrections, and bounded signal-transition effects gives a finite $K_V$ such that
\[
    \mathcal L_{\bm s}V_k(t)
    \le
    (\bar{\bm w}^{k}(t))^\top\bar{\bm\lambda}
    -
    \Psi_k(\bm s,t)
    +
    K_V.
\]
This proves~\eqref{eq:service_weighted_drift}.

\subsection{Proof of Proposition~\ref{prop:bounded_error_mp}}
\label{app:proof_bounded_error_mp}

The proof compares the implemented decentralized phase decision with the single network-consistent analytical objective in three steps: intersection-wise composition, maximum-red restriction, and within-interval variation. For $a\in\mathcal M_n$, define the local proof-equivalent movement pressure
\[
    p_{n,a}^{\theta}(t)
    \triangleq
    \sum_h\theta_{n,h,k}p_a^h(t).
\]
Since $\nu(a)=n$, its upstream term is identical to that of~\eqref{eq:analytical_pressure}, and
\begin{equation}
    p_{n,a}^{\theta}(t)-\bar p_a^{k}(t)
    =
    \sum_{b\in\mathcal D(a)\cap\mathcal M_{\mathcal N}}
    r_b(t)
    \sum_h
    \left(
        \theta_{\nu(b),h,k}-\theta_{n,h,k}
    \right)
    w_b^h(t).
    \label{eq:app_pressure_discrepancy}
\end{equation}
For any real movement $c\in\mathcal M_{\mathcal N}$, define $W_c\triangleq\max_hW_c^h$, so that $W_a$ and $W_b$ below denote the corresponding constituent-state bounds for movements $a$ and $b$. Since the $\ell_1$ distance between two simplex vectors is at most two,
\begin{equation}
    \left|
        p_{n,a}^{\theta}(t)-\bar p_a^{k}(t)
    \right|
    \le
    2
    \sum_{b\in\mathcal D(a)\cap\mathcal M_{\mathcal N}}
    r_b(t)W_b
    \le
    2
    \sum_{b\in\mathcal M_{\mathcal N}}W_b
    \triangleq
    E_a.
    \label{eq:app_pressure_error_bound}
\end{equation}
Here, $E_a$ is a finite uniform bound on the local-to-network pressure discrepancy for movement $a$.

Define the network sum of local proof-equivalent objectives
\[
    \Phi_k^{L}(\bm s,t)
    \triangleq
    \sum_{n\in\mathcal N}
    \sum_{a\in\mathcal M_n}
    p_{n,a}^{\theta}(t)\bar d_a(\bm s,t).
\]
The superscript $L$ indicates that $\Phi_k^{L}$ is assembled from these local proof-equivalent objectives. Using~\eqref{eq:app_pressure_error_bound} and $0\le\bar d_a\le c_a$,
\begin{equation}
    \left|
        \Phi_k^{L}(\bm s,t)
        -
        \Psi_k(\bm s,t)
        +
        C_{\mathcal F}(t)
    \right|
    \le
    \sum_{a\in\mathcal M_{\mathcal N}}c_aE_a
    \triangleq
    B_\theta,
    \label{eq:app_objective_composition_error}
\end{equation}
where $B_\theta$ is the resulting finite objective-level composition-discrepancy bound and $C_{\mathcal F}(t)\triangleq\sum_{a\in\mathcal M_{\mathcal F}}\bar p_a^{k}(t)\bar d_a(t)$ is independent of the candidate real-intersection phases. Equivalently, $C_{\mathcal F}(t)$ cancels from every comparison between two feasible real-intersection signal states.

At $t_k$, Proposition~\ref{prop:normalized_composition_equivalence} implies that the implemented phase maximizes the corresponding local proof-equivalent objective over $\mathcal P_n^{\mathrm{adm}}(t_k)$. For any constituent,
\begin{equation}
    |p_a^h(t_k)|
    \le
    W_a
    +
    \sum_{b\in\mathcal D(a)\cap\mathcal M_{\mathcal N}}
    r_b(t_k)W_b
    \le
    W_a
    +
    \sum_{b\in\mathcal M_{\mathcal N}}W_b
    \triangleq
    \bar P_a.
    \label{eq:app_constituent_pressure_bound}
\end{equation}
Because $\bm\theta_{n,k}\in\Delta_H$,
$|p_{n,a}^{\theta}(t_k)|\le\bar P_a$, where $\bar P_a$ is the finite bound on the magnitude of the constituent pressure for movement $a$ defined in~\eqref{eq:app_constituent_pressure_bound}. Hence every local raw phase objective satisfies
\begin{equation}
    \left|
        \sum_{a\in\mathcal M_n}
        p_{n,a}^{\theta}(t_k)
        \bar d_a(\bm s_n^\phi,t_k)
    \right|
    \le
    \sum_{a\in\mathcal M_n}c_a\bar P_a
    \triangleq
    B_n^{\mathrm{loc}}.
    \label{eq:app_local_objective_bound}
\end{equation}
Here, $B_n^{\mathrm{loc}}$ is a finite bound on the magnitude of the local proof-equivalent phase objective at intersection $n$. Therefore restricting the maximization from $\mathcal P_n$ to the nonempty $\mathcal P_n^{\mathrm{adm}}(t_k)$ can reduce the local proof-equivalent objective by at most $B_{\mathrm{red},n}\triangleq2B_n^{\mathrm{loc}}$, where $B_{\mathrm{red},n}$ bounds the loss caused by the maximum-red phase restriction at intersection $n$. Summing over intersections gives the finite network-level bound $B_{\mathrm{red}}\triangleq\sum_nB_{\mathrm{red},n}<\infty$, and
\begin{equation}
    \Phi_k^{L}(\bm s_k^*,t_k)
    \ge
    \max_{\bm s\in\mathcal S}
    \Phi_k^{L}(\bm s,t_k)
    -
    B_{\mathrm{red}}.
    \label{eq:app_maxred_local_gap}
\end{equation}
Combining~\eqref{eq:app_objective_composition_error} and~\eqref{eq:app_maxred_local_gap},
\begin{equation}
    \Psi_k(\bm s_k^*,t_k)
    \ge
    \max_{\bm s\in\mathcal S}
    \Psi_k(\bm s,t_k)
    -
    2B_\theta-B_{\mathrm{red}}.
    \label{eq:app_epoch_approximate_mp}
\end{equation}

For the within-interval comparison, write the analytical objective as
\[
    \Psi_k(\bm s,t)
    =
    \Psi_k^{\mathcal N}(\bm s,t)
    +
    C_{\mathcal F}(t),
\]
where
\[
    \Psi_k^{\mathcal N}(\bm s,t)
    \triangleq
    \sum_{a\in\mathcal M_{\mathcal N}}
    \bar p_a^{k}(t)\bar d_a(\bm s,t)
\]
denotes the real-movement component of the analytical objective. The source term $C_{\mathcal F}(t)$ is independent of $\bm s$ and therefore cancels from the maximization. Since all constituent states on real movements are uniformly bounded and all candidate real-movement discharges are bounded, there exists a finite $B_\Delta$ such that
\begin{equation}
    \left|
        \Psi_k^{\mathcal N}(\bm s,t)
        -
        \Psi_k^{\mathcal N}(\bm s,t_k)
    \right|
    \le
    B_\Delta,
    \qquad
    \forall\bm s\in\mathcal S,\quad t\in I_k.
    \label{eq:app_interval_objective_variation}
\end{equation}
Applying this bound to the selected and maximizing signal states yields
\begin{equation}
    \Psi_k(\bm s_k^*,t)
    \ge
    \max_{\bm s\in\mathcal S}
    \Psi_k(\bm s,t)
    -
    2B_\theta-B_{\mathrm{red}}-2B_\Delta.
    \label{eq:app_interval_approximate_mp}
\end{equation}
Thus Proposition~\ref{prop:bounded_error_mp} holds with
\[
    B=2B_\theta+B_{\mathrm{red}}+2B_\Delta<\infty.
\]

\subsection{Proof of Lemma~\ref{lem:fixed_action_drift}}
\label{app:proof_fixed_action_drift}

Since $\bar{\bm\lambda}\in\bm\Lambda^\circ$, Definition~\ref{def:strict_stabilizable_region} provides an admissible reference policy $\sigma^{\mathrm{ref}}$, a matched Lyapunov functional $V^{\mathrm{ref}}$, a uniformly bounded residual $\xi^{\mathrm{ref}}$, and finite constants $K_{\mathrm{ref}}$, $K_0$, and $\epsilon>0$. By~\eqref{eq:reference_matched_drift},
\[
    (\bm w^{\mathrm{ref}}(t))^\top\bar{\bm\lambda}
    -
    \Psi^{\mathrm{ref}}(\bm s^{\mathrm{ref}}(t),t)
    =
    \mathcal L_{\sigma^{\mathrm{ref}}}V^{\mathrm{ref}}(t)
    -
    \xi^{\mathrm{ref}}(t).
\]
Combining this identity with the strict drift condition and $|\xi^{\mathrm{ref}}(t)|\le K_{\mathrm{ref}}$ gives
\begin{equation}
    (\bm w^{\mathrm{ref}}(t))^\top\bar{\bm\lambda}
    -
    \Psi^{\mathrm{ref}}(\bm s^{\mathrm{ref}}(t),t)
    \le
    K_0+K_{\mathrm{ref}}
    -
    \epsilon
    \sum_{a\in\mathcal M_{\mathcal F}}z_a(t).
    \label{eq:app_reference_objective_base}
\end{equation}
The reference and composite states use the same physical source queues, and $\bar{\bm\lambda}$ has zero entries on real movements. Hence
\[
    (\bm w^{\mathrm{ref}}(t))^\top\bar{\bm\lambda}
    =
    (\bar{\bm w}^{k}(t))^\top\bar{\bm\lambda}.
\]
For the same reference action, the difference between $\Psi^{\mathrm{ref}}$ and $\Psi_k$ contains only real-movement state terms. These states and candidate discharges are uniformly bounded, so there exists $K_{\mathrm{cmp}}<\infty$ such that
\begin{equation}
    \left|
        \Psi^{\mathrm{ref}}(\bm s^{\mathrm{ref}}(t),t)
        -
        \Psi_k(\bm s^{\mathrm{ref}}(t),t)
    \right|
    \le
    K_{\mathrm{cmp}}.
    \label{eq:app_reference_objective_difference}
\end{equation}
Therefore
\begin{equation}
    (\bar{\bm w}^{k}(t))^\top\bar{\bm\lambda}
    -
    \Psi_k(\bm s^{\mathrm{ref}}(t),t)
    \le
    K_0+K_{\mathrm{ref}}+K_{\mathrm{cmp}}
    -
    \epsilon
    \sum_{a\in\mathcal M_{\mathcal F}}z_a(t).
    \label{eq:app_reference_objective_bound}
\end{equation}
Proposition~\ref{prop:bounded_error_mp} gives
\[
    \Psi_k(\bm s_k^*,t)
    \ge
    \Psi_k(\bm s^{\mathrm{ref}}(t),t)-B.
\]
Combining this relation with~\eqref{eq:service_weighted_drift} gives
\[
    \mathcal L_{\bm s_k^*}V_k(t)
    \le
    K_0+K_{\mathrm{ref}}+K_{\mathrm{cmp}}+B+K_V
    -
    \epsilon
    \sum_{a\in\mathcal M_{\mathcal F}}z_a(t),
\]
which proves~\eqref{eq:fixed_action_drift}.

\subsection{Proof of Lemma~\ref{lem:bounded_composition_jump}}
\label{app:proof_composition_jump}

At a decision epoch, the turning-ratio estimate and composition action may both be updated, but the physical traffic state and the constituent Lyapunov terms are unchanged. The real-movement quadratic term and all source terms in~\eqref{eq:composite_lyapunov} are independent of the analytical coefficients. Hence
\begin{align}
    V_k(t_k^+)-V_{k-1}(t_k^-)
    =
    \sum_{a\in\mathcal M_{\mathcal N}}
    \sum_h
    \left(
        \theta_{\nu(a),h,k}
        -
        \theta_{\nu(a),h,k-1}
    \right)
    G_a^h(t_k).
    \label{eq:app_switch_jump}
\end{align}
Let $G_a^{\max}\triangleq\max_hG_a^{h,\max}$. Since the $\ell_1$ distance between two simplex vectors is at most two,
\begin{align}
    \left|
        V_k(t_k^+)-V_{k-1}(t_k^-)
    \right|
    &\le
    2
    \sum_{a\in\mathcal M_{\mathcal N}}
    G_a^{\max}
    \nonumber\\
    &\triangleq
    K_{\mathrm{sw}}
    <
    \infty.
    \label{eq:app_switch_bound}
\end{align}

\subsection{Completion of Theorem~\ref{thm:com_mp_stability}}
\label{app:completion_com_mp_stability}

For any integer $m\ge1$, let $t_m=m\Delta$ denote the end of the first $m$ complete decision intervals. Applying the extended drift identity to~\eqref{eq:fixed_action_drift} over each interval and adding the update jumps bounded by~\eqref{eq:app_switch_bound} gives
\begin{align}
    \mathbb E[V_m(t_m^+)]
    -
    \mathbb E[V_0(t_0^+)]
    \le{}&
    m(K\Delta+K_{\mathrm{sw}})
    -
    \epsilon
    \int_0^{t_m}
    \mathbb E
    \left[
        \sum_{a\in\mathcal M_{\mathcal F}}z_a(t)
    \right]dt.
    \label{eq:app_telescoping}
\end{align}
Because $V_m(t_m^+)\ge0$ and $t_m=m\Delta$,
\begin{equation}
    \frac{1}{t_m}
    \int_0^{t_m}
    \mathbb E
    \left[
        \sum_{a\in\mathcal M_{\mathcal F}}z_a(t)
    \right]dt
    \le
    \frac{K}{\epsilon}
    +
    \frac{K_{\mathrm{sw}}}{\epsilon\Delta}
    +
    \frac{\mathbb E[V_0(t_0^+)]}{\epsilon t_m}.
    \label{eq:app_source_average}
\end{equation}
Taking $m\to\infty$ proves a finite long-run average source queue. The same bound extends to arbitrary $T$ because the final partial interval has length less than $\Delta$. Real-movement counts are uniformly bounded by finite storage, completing the proof.

\subsection{Proof of Proposition~\ref{prop:qcv_admissibility}}
\label{app:bounded_cv_travel_time}

We first verify the Lyapunov-term drift relations of the two representative constituents.

For Q-MP, the state and matched Lyapunov term are
\[
    w_a^Q=\frac{z_a}{L_i},
    \qquad
    G_a^Q=\frac{z_a^2}{2L_i}.
\]
For any vehicle-transition event, write $\Delta_{\mathrm e}X=X^+-X^-$ for the event-induced jump of a quantity $X$, where $X^-$ and $X^+$ denote its pre-event and post-event values, respectively; this notation is distinct from the decision-interval duration $\Delta$. Under an absolutely continuous count evolution,
\[
    \frac{dG_a^Q}{dt}
    =
    w_a^Q\dot z_a.
\]
Let $N$ denote the pre-event vehicle count on movement $a$. Under a unit vehicle arrival from state $z_a=N$,
\[
    \Delta_{\mathrm e} G_a^Q
    =
    \frac{(N+1)^2-N^2}{2L_i}
    =
    w_a^Q
    +
    \frac{1}{2L_i},
\]
while under a unit departure,
\[
    \Delta_{\mathrm e} G_a^Q
    =
    \frac{(N-1)^2-N^2}{2L_i}
    =
    -w_a^Q
    +
    \frac{1}{2L_i}.
\]
Hence Q-MP satisfies~\eqref{eq:constituent_potential_drift} with a finite remainder constant.

For CV-MP, use
\[
    w_a^{CV}
    =
    \sum_{v\in\mathcal J_a^{CV}}
    \frac{\tau_v}{\bar\tau_a},
    \qquad
    G_a^{CV}=z_aw_a^{CV}.
\]
Between movement-transition events, $z_a$ and the connected-vehicle set are fixed. Using $|\mathcal J_a^{CV}|$ for the number of CVs currently contributing to $w_a^{CV}$,
\begin{equation}
    \dot G_a^{CV}
    =
    z_a\dot w_a^{CV}
    =
    \frac{z_a|\mathcal J_a^{CV}|}{\bar\tau_a}
    \le
    \frac{(z_a^{\max})^2}{\bar\tau_a}.
    \label{eq:app_cv_between_event}
\end{equation}
Now consider a vehicle departure from movement $a$. Let the pre-event state be $(z_a,w_a^{CV})=(N,W)$. If the departing vehicle is a CV, let $\vartheta_v=\tau_v/\bar\tau_a$, and otherwise set $\vartheta_v=0$. Then
\begin{align}
    \Delta_{\mathrm e} G_a^{CV}
    &=
    (N-1)(W-\vartheta_v)-NW
    \nonumber\\
    &=
    -W-(N-1)\vartheta_v
    \le
    -W
    =
    w_a^{CV}\Delta_{\mathrm e} z_a.
    \label{eq:app_cv_departure}
\end{align}
For a vehicle entering movement $a$, the elapsed-time clock on the new movement starts from zero. If the pre-event state is $(N,W)$,
\begin{equation}
    \Delta_{\mathrm e} G_a^{CV}
    =
    (N+1)W-NW
    =
    W
    =
    w_a^{CV}\Delta_{\mathrm e} z_a.
    \label{eq:app_cv_arrival}
\end{equation}
Together with~\eqref{eq:app_cv_between_event}, these event relations establish~\eqref{eq:constituent_potential_drift} for CV-MP with a finite upper remainder constant determined by finite storage and $\bar\tau_a$.

We next establish the uniform real-movement bounds used by the adaptive composition. Consider a CV $v$ entering real movement $a$. The non-control running component of its link travel time is uniformly bounded by Assumption~\ref{ass:recoverable_cv}. By finite storage and FIFO, at most $z_a^{\max}$ vehicles can be ahead of $v$. The maximum-red rule guarantees a serving activation within a uniformly bounded time whenever $a$ remains occupied, so $m_a$ serving activations occur within a finite horizon determined by $m_a$, $T_{\nu(a),\phi(a)}^{red,\max}$, and the decision interval $\Delta$. By Assumption~\ref{ass:recoverable_cv}, those $m_a$ activations discharge at least $\underline d_a$ vehicles whenever $a$ remains occupied.

Define $N_a=\left\lceil z_a^{\max}/\underline d_a\right\rceil+1$. After at most $N_a$ blocks of $m_a$ serving-phase activations, the cumulative discharge exceeds the maximum vehicle mass that can remain ahead of $v$, so $v$ must cross the stop line under FIFO. Hence its control-delay component is bounded by a finite constant depending on $N_a$, $m_a$, $T_{\nu(a),\phi(a)}^{red,\max}$, and $\Delta$, with bounded phase-transition time absorbed into the same constant. Adding the bounded running component gives a finite uniform link-travel-time bound $\tau_a^{\max}$.

It follows that
\[
    w_a^{CV}(t)
    \le
    z_a^{\max}
    \frac{\tau_a^{\max}}{\bar\tau_a}
    =
    W_a^{CV}
    <\infty,
\]
and
\[
    G_a^{CV}(t)
    =
    z_a(t)w_a^{CV}(t)
    \le
    z_a^{\max}W_a^{CV}
    <\infty.
\]
These bounds provide the finite constants required for the local-to-network composition bound in Proposition~\ref{prop:bounded_error_mp} and the analytical-composition update bound in Lemma~\ref{lem:bounded_composition_jump}. Together with the event-level drift relations above, the Q-MP and CV-MP states satisfy Assumption~\ref{ass:constituent_admissibility}.

\begin{figure}
    \centering
    \includegraphics[width=0.5\textwidth]{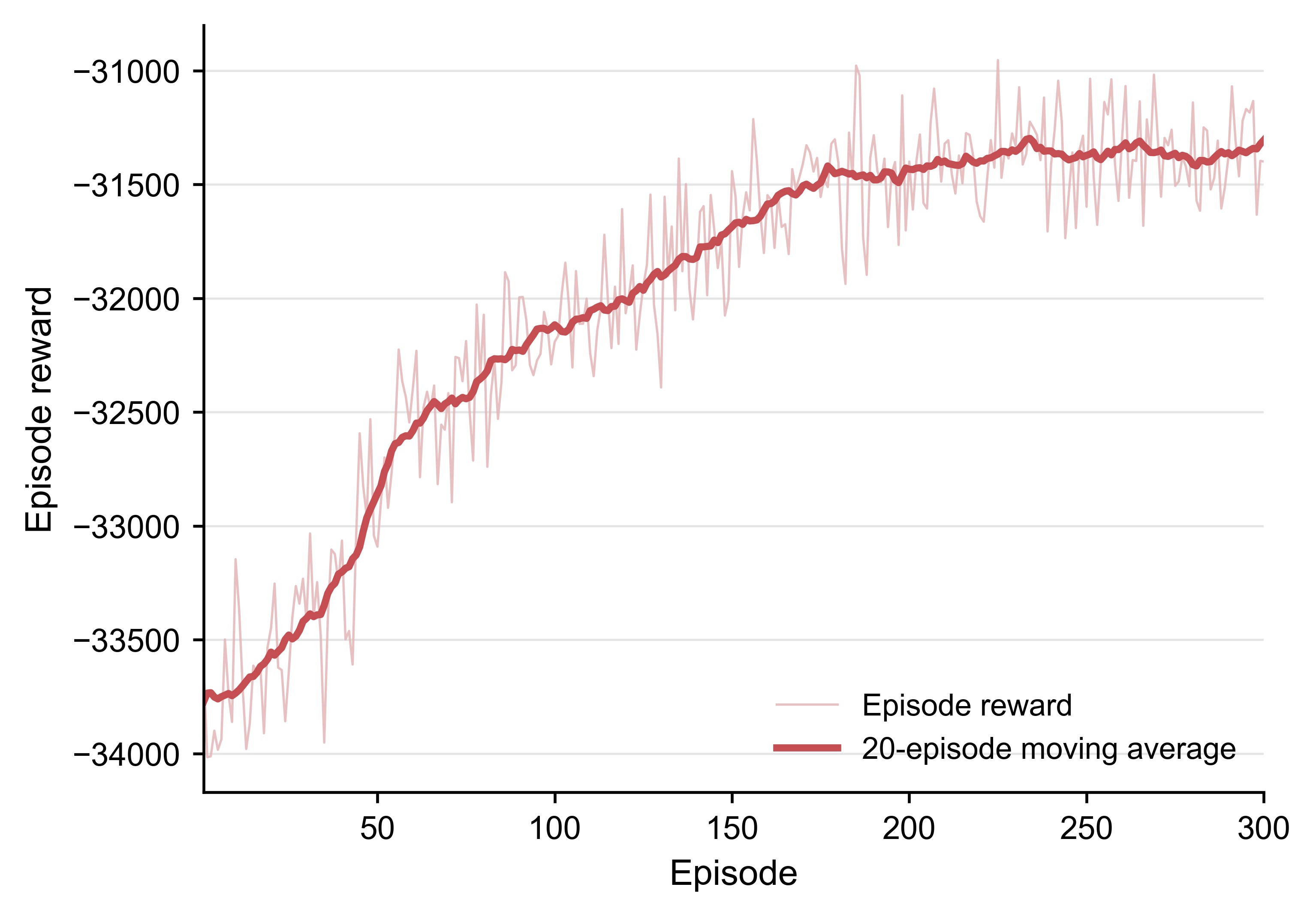}
    \caption{Training reward of Com-MP.}
    \label{FIG:train_reward}
\end{figure}

\begin{figure}
    \centering
    \includegraphics[width=0.9\textwidth]{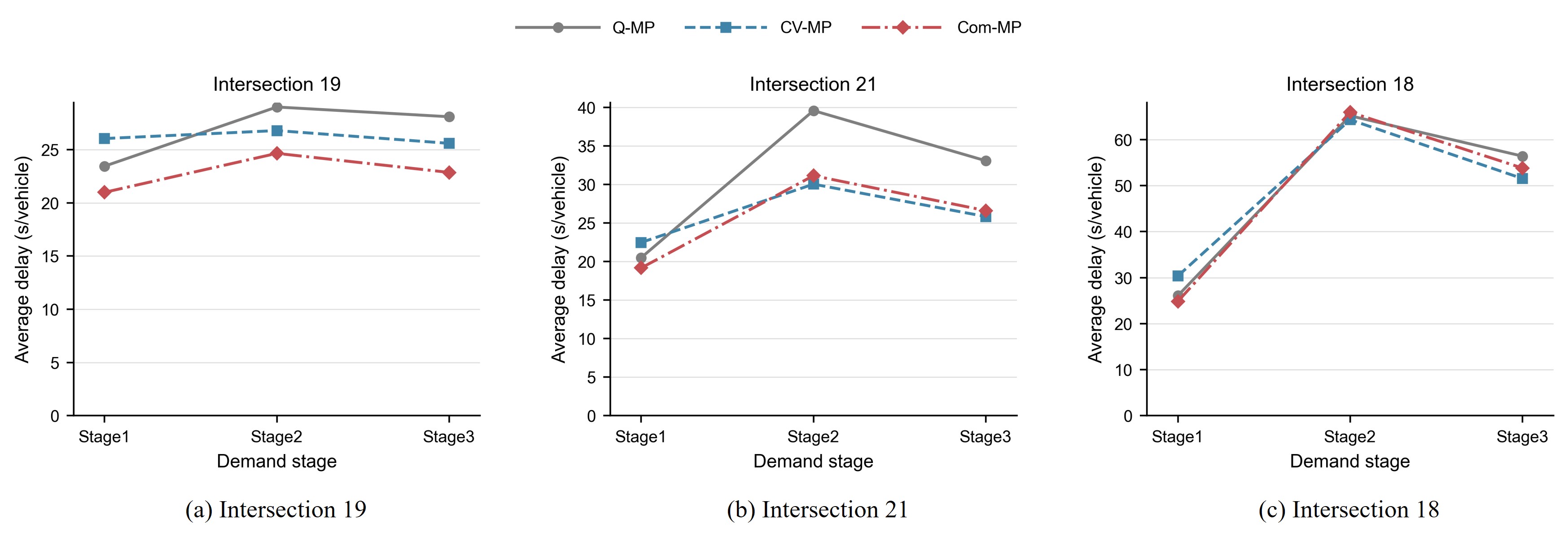}
    \caption{Delay performance at representative intersections across the three traffic stages.}
    \label{FIG:three_intersections}
\end{figure}

\section{Training convergence of Com-MP}
\label{app:training_convergence}

Figure~\ref{FIG:train_reward} reports the episode reward and its 20-episode moving average during training. The reward increases rapidly in the early stage and then improves more gradually before approaching a stable level, indicating convergence of the learned composition policy under the adopted training configuration.

\section{Representative intersection-level performance patterns}
\label{app:representative_intersections}

To complement the intersection-level heterogeneity analysis in Section~\ref{sec:intersection_heterogeneity}, Fig.~\ref{FIG:three_intersections} illustrates three representative local performance patterns across the three traffic stages. At Intersection~19, Q-MP performs better than CV-MP in Stage~1, whereas CV-MP performs better in Stages~2 and~3; Com-MP achieves the lowest delay in all three stages. At Intersection~21, the better constituent also changes across stages, while Com-MP is best in Stage~1 and remains close to CV-MP thereafter. At Intersection~18, Com-MP performs best in Stage~1, whereas CV-MP is slightly better in the later stages. These examples further illustrate that the locally preferred constituent is not fixed across locations or traffic conditions and that the benefit of adaptive composition is correspondingly heterogeneous.

\bibliographystyle{cas-model2-names}
\bibliography{citation}

\end{document}